\documentclass[twoside, journal, 11pt, onecolumn]{IEEEtran}

\usepackage{amsmath}
\usepackage{amsthm}
\usepackage{amssymb}
\usepackage{epsfig}
\usepackage{graphics}
\usepackage{graphicx}
\usepackage{float}
\usepackage{url}

\begin{document}

\input epsf

\def\thesection {\arabic{section}}

\newtheorem{corollary}{Corollary}
\newtheorem{theorem}{Theorem}
\newtheorem{lemma}{Lemma}
\newtheorem{remark}{Remark}
\newtheorem{example}{Example}
\newtheorem{proposition}{Proposition}
\newtheorem{question}{Question}
\newtheorem{conjecture}{Conjecture}

\newcommand{\expectation}{\ensuremath{\mathbb{E}}}
\newcommand{\reals}{\ensuremath{\mathbb{R}}}

\title{\huge{On the Corner Points of the Capacity Region of a Two-User Gaussian
Interference Channel}}

\markboth{Revised in April 5, 2015}{On the Corner Points of
the Capacity Region of a Gaussian Interference Channel}

\author{Igal Sason\thanks{The author is with the Department of Electrical Engineering,
Technion, Haifa 32000, Israel. E-mail: {\tt sason@ee.technion.ac.il}. The paper was submitted
to the {\em IEEE Trans. on Information Theory} in July~18, 2014, and revised in
April~5, 2015. It was presented in part at the {\em 2014 IEEE International Symposium
on Information Theory (ISIT~2014)}, Honolulu, Hawaii, USA, July 2014.
This research work has been supported by the Israeli Science Foundation (ISF), grant number~12/12.}}

\maketitle

\vspace*{-0.9cm}
\begin{abstract}
This work considers the corner points of the capacity region of a two-user Gaussian interference
channel (GIC). In a two-user GIC, the rate pairs where one user transmits its data at the
single-user capacity (without interference), and the other at the largest rate for which
reliable communication is still possible are called corner points. This paper relies on existing outer
bounds on the capacity region of a two-user GIC that are used to derive informative bounds on the corner points
of the capacity region. The new bounds refer to a weak two-user GIC (i.e., when both cross-link
gains in standard form are positive and below~1), and a refinement of these bounds is obtained for the case
where the transmission rate of one user is within $\varepsilon > 0$ of the single-user capacity. The bounds
on the corner points are asymptotically tight as the transmitted powers tend to infinity, and they are
also useful for the case of moderate SNR and INR.
Upper and lower bounds on the gap (denoted by $\Delta$) between the sum-rate and the maximal achievable
total rate at the two corner points are derived. This is followed by an asymptotic analysis analogous
to the study of the generalized degrees of freedom (where the SNR and INR scalings are coupled such that
$\frac{\log(\text{INR})}{\log(\text{SNR})} = \alpha \geq 0$), leading to an asymptotic
characterization of this gap which is exact for the whole range of $\alpha$. The upper and lower bounds
on $\Delta$ are asymptotically tight in the sense that they achieve the exact asymptotic characterization.
Improved bounds on $\Delta$ are derived for finite SNR and INR, and their improved tightness is exemplified numerically.
\end{abstract}


\vspace*{-0.3cm}
\section{Introduction}
\label{section: Introduction}
The two-user Gaussian interference channel (GIC) has been extensively studied
in the literature during the last four decades (see, e.g.,
\cite[Chapter~6]{ElGamalKim_book}, \cite{FnT13} and references therein). For completeness and
to set notation, the model of a two-user GIC in  {\em standard form} is introduced
shortly: this discrete-time, memoryless interference channel
is characterized by the following equations that relate the paired inputs
$(X_1, X_2)$ and outputs $(Y_1, Y_2)$:
\begin{align}
Y_1 = X_1 + \sqrt{a_{12}} \, X_2 + Z_1 \label{eq:Gaussian-IC1} \\
Y_2 = \sqrt{a_{21}} \, X_1 + X_2 + Z_2 \label{eq:Gaussian-IC2}
\end{align}
where the cross-link gains $a_{12}$ and $a_{21}$ are time-invariant,
the inputs and outputs are real valued, and $Z_1$ and $Z_2$ denote
additive real-valued Gaussian noise samples. Let
$X_1^n \triangleq (X_{1,1}, \ldots, X_{1,n})$ and
$X_2^n \triangleq (X_{2,1}, \ldots, X_{2,n})$ be two transmitted
codewords across the channel where $X_{i,j}$ denotes the symbol that is
transmitted by user $i$ at time instant $j$ (here, $i \in \{1, 2\}$ and
$j \in \{1, \ldots, n\}$). The model assumes that there is no cooperation
between the transmitters, not between the receivers. It is assumed, however,
that the receivers have full
knowledge of the codebooks used by both users. The power constraints on the inputs
are given by $\frac{1}{n} \sum_{j=1}^n \expectation[X_{1,j}^2] \leq P_1$
and $\frac{1}{n} \sum_{j=1}^n \expectation[X_{2,j}^2] \leq P_2$ where
$P_1, P_2 > 0$. The random vectors $Z_1^n$ and $Z_2^n$
have i.i.d. Gaussian entries with zero mean and unit variance, and
they are independent of the inputs $X_1^n$ and $X_2^n$. Furthermore, $Z_1^n$
and $Z_2^n$ can be assumed to be independent since the capacity region of a
two-user, discrete-time, memoryless interference channel depends only on the
marginal {\em pdf}s $p(y_i | x_1, x_2)$ for $i \in \{1, 2\}$ (as the receivers
do not cooperate). Finally,
perfect synchronization between the pairs of transmitters and receivers
allows time-sharing between the users, which implies that the capacity region is convex.

Depending on the values of $a_{12}$ and $a_{21}$, the two-user GIC is classified
into weak, strong, mixed, one-sided and degraded GIC. If $0 < a_{12}, a_{21} < 1$,
the channel is called a {\em weak} GIC. If $a_{12} \geq 1$ and $a_{21} \geq 1$, the
channel is a {\em strong} GIC; furthermore, if $a_{12} \geq 1+P_1$ and $a_{21} \geq 1+P_2$
then the channel is a {\em very strong} GIC, and its capacity region is not harmed
(i.e., reduced) as a result of the interference \cite{Carleial75}.
If either $a_{12} \geq 1$ and $0 < a_{21} < 1$ or $a_{21} \geq 1$ and $0 < a_{12} < 1$,
the channel is called a {\em mixed} GIC; the special case where $a_{12} a_{21} = 1$ is
called a {\em degraded} GIC. It is a {\em one-sided} GIC if either $a_{12}=0$ or $a_{21}=0$;
a one-sided GIC is either weak or strong if its non-zero cross-link gain is below or
above~1, respectively.
Finally, a {\em symmetric} GIC refers to the case where $a_{12} = a_{21}$ and $P_1=P_2$.

In spite of the simplicity of the model of a two-user GIC, the exact characterization of
its capacity region is yet unknown, except for strong
(\cite{Han81}, \cite{Sato81}) or very strong interference \cite{Carleial75}.
For other GICs, not only the capacity region is yet unknown but even its corner
points are not fully determined. For
mixed or one-sided GICs, a single corner point of the capacity region is known
and it attains the sum-rate of this channel (see \cite[Section~6.A]{Motahari09},
\cite[Theorem~2]{Sason04}, and \cite[Section~2.C]{Kramer09}). For weak GICs,
both corner points of the capacity region are yet unknown.

The operational meaning of the study of the corner points of the capacity
region for a two-user GIC is to explore the situation where one
transmitter sends its information at the maximal achievable rate for a
single user (in the absence of interference), and the second transmitter
maintains a data rate that enables reliable communication to the two
non-cooperating receivers \cite{Costa85}.
Two questions occur in this scenario:
\begin{question}
{\em What is the maximal achievable rate of the second transmitter ?}
\label{question 1: corner points}
\end{question}
\begin{question}
{\em Does it enable the first receiver to reliably decode the messages of both transmitters ?}
\label{question 2: corner points}
\end{question}

In his paper \cite{Costa85}, Costa presented an approach suggesting that when one of the
transmitters, say transmitter~1, sends its data over a two-user GIC at the maximal interference-free
rate $R_1 = \frac{1}{2} \log(1+P_1)$ bits per channel use, then the maximal rate $R_2$ of
transmitter~2 is the rate that enables receiver~1 to decode both messages.
The corner points of the capacity region are therefore related to a multiple-access channel
where one of the receivers decodes correctly both messages. However, \cite[pp.~1354--1355]{Sason04}
pointed out a gap in the proof of \cite[Theorem~1]{Costa85}, though it was conjectured that the
main result holds. It therefore leads to the following conjecture:
\begin{conjecture}
{\em For rate pairs $(R_1, R_2)$ in the capacity region of a two-user GIC with arbitrary
positive cross-link gains $a_{12}$ and $a_{21}$, and power constraints $P_1$ and $P_2$, let
\begin{equation}
C_1 \triangleq \frac{1}{2} \, \log(1+P_1), \quad
C_2 \triangleq \frac{1}{2} \, \log(1+P_2)
\label{eq:capacity of AWGNs}
\end{equation}
be the capacities of the single-user AWGN channels (in the absence of interference), and let
\begin{align}
& R_1^* \triangleq \frac{1}{2} \, \log \left(1 + \frac{a_{21} P_1}{1+P_2} \right)
\label{eq:maximal R1 in Costa's conjecture} \\
& R_2^* \triangleq \frac{1}{2} \, \log \left(1 + \frac{a_{12} P_2}{1+P_1} \right).
\label{eq:maximal R2 in Costa's conjecture}
\end{align}
Then, the following is conjectured to hold for achieving reliable communication at both receivers:
\begin{enumerate}
\item If $R_2 \geq C_2 - \varepsilon$, for an arbitrary $\varepsilon > 0$, then
$ R_1 \leq R_1^* + \delta_1(\varepsilon)$
where $\delta_1(\varepsilon) \rightarrow 0$ as $\varepsilon \rightarrow 0$.
\item If $R_1 \geq C_1 - \varepsilon$, then $ R_2 \leq R_2^* + \delta_2(\varepsilon)$
where $\delta_2(\varepsilon) \rightarrow 0$ as $\varepsilon \rightarrow 0$.
\end{enumerate}}
\label{conjecture: Costa's conjecture for the corner points}
\end{conjecture}

The discussion on Conjecture~\ref{conjecture: Costa's conjecture for the corner points}
is separated in the continuation to this section into mixed,
strong, and weak one-sided GICs. This is done by restating some known results
from \cite{Costa85}, \cite{Han81}, \cite{Motahari09}, \cite{Sason04},
\cite{Sato78}, \cite{Sato81} and \cite{Kramer09}. The focus of this paper is
on weak GICs. For this class,
the corner points of the capacity region are yet unknown, and they are
studied in the converse part of this paper by relying
on some existing outer bounds on the capacity region.
Various outer bounds
on the capacity region of GICs that have been introduced in the literature
(see, e.g., \cite{Annapureddy_Veeravalli09}, \cite{Carleial78}, \cite{ElGamalKim_book},
\cite{EtkinTseWang08}, \cite{Kramer04}, \cite{Motahari09}, \cite{Nam_Caire_ISIT2012},
\cite{Sato78} and \cite{Kramer09}--\cite{Tuninetti_ISIT2011}).
The analysis in this paper
provides informative bounds that are given in closed form, and they are
asymptotically tight for sufficiently large SNR and INR. Improvements of these
bounds are derived for finite SNR and INR, and these improvements
are exemplified numerically.

\subsection{On Conjecture 1 for Mixed GICs}
\label{subsection: On Conjecture 1 for the Gaussian IC with mixed interference}
Conjecture~\ref{conjecture: Costa's conjecture for the corner points}
is considered in the following for mixed GICs:
\begin{proposition}
{\em Consider a mixed GIC where $a_{12} \geq 1$ and
$a_{21} < 1$, and assume that transmitter~1 sends its message at rate
$R_1 \geq C_1 - \varepsilon$ for an arbitrary $\varepsilon > 0$.
Then, the following holds:
\begin{enumerate}
\item If $1-a_{12} < (a_{12} a_{21}-1) P_1$, then
$R_2 \leq \frac{1}{2} \, \log\left(1 + \frac{P_2}{1+a_{21} P_1}\right) + \varepsilon.$
This implies that the maximal rate $R_2$ is {\em strictly smaller} than the corresponding upper bound
in Conjecture~\ref{conjecture: Costa's conjecture for the corner points}.
\item Otherwise, if $1-a_{12} \geq (a_{12} a_{21}-1) P_1$, then
$R_2 \leq R_2^* + \varepsilon$.
This coincides with the upper bound in Conjecture~\ref{conjecture: Costa's conjecture for the corner points}.
\end{enumerate}
The above two items refer to a corner point that achieves the sum-rate. On the other hand,
if $R_2 \geq C_2 - \varepsilon$, then
\begin{align}
R_1 \leq \frac{1}{2} \log\left(1 + \frac{P_1}{1+P_2}\right) + \delta(\varepsilon)
\label{eq:maximal R1 for a two-user mixed Gaussian IC}
\end{align}
where $\delta(\varepsilon) \rightarrow 0$ as $\varepsilon \rightarrow 0$.}
\label{proposition: mixed interference}
\end{proposition}

\begin{proof}
The first two items of this proposition follow
from \cite[Theorem~10]{Motahari09} or the earlier result in
\cite[Theorem~1]{Kramer04}. Eq.~\eqref{eq:maximal R1 for a two-user mixed Gaussian IC}
is a consequence of \cite[Theorem~2]{Kramer04}.
\end{proof}

\subsection{On Conjecture 1 for Strong GICs}
\label{subsection: On conjecture 1 for the Gaussian IC with strong interference}
The capacity region of a strong GIC is equal to the
intersection of the capacity regions of the two Gaussian
multiple-access channels from the two transmitters to each one of
the receivers (see \cite[Theorem~5.2]{Han81} and \cite{Sato81}). The
two corner points of this capacity region are consistent with
Conjecture~\ref{conjecture: Costa's conjecture for the corner points}.
Question~\ref{question 2: corner points} is answered in
the affirmative for a strong GIC because each receiver
is able to decode the messages of both users.

The capacity region of a very strong GIC, where $a_{12} \geq 1+P_1$ and $a_{21} \geq 1+P_2$,
is not affected by the interference \cite{Carleial75}. This is a trivial case where
Conjecture~\ref{conjecture: Costa's conjecture for the corner points} does not
provide a tight upper bound on the maximal transmission rate (note that if
$a_{12} > 1+P_1$ and $a_{21} > 1+P_2$, then $R_1^* > C_1$ and $R_2^* > C_2$).

\subsection{On the Corner Points of Weak One-Sided GICs}
\label{subsection: On Conjecture 1 for the one-sided Gaussian IC}
In \cite{Costa85}, an interesting equivalence has been established between weak one-sided GICs
and degraded GICs: a weak one-sided GIC with power constraints $P_1$ and $P_2$, and cross-link
gains $a_{12}=0$ and $a_{21}=a \in (0,1)$ in standard form, has an identical capacity
region to that of a degraded GIC whose standard form is given by
\begin{equation}
Y_1 = X_1 + \sqrt{\frac{1}{a}} \, X_2 + Z_1, \quad Y_2 = \sqrt{a}  X_1 + X_2 + Z_2
\label{eq:Y1,2 for degraded Gaussian IC}
\end{equation}
with the same power constraints on the inputs, and where
$Z_1 Z_2$ are independent Gaussian random variables with zero mean and unit variance.
The first part of Proposition~\ref{proposition: mixed interference} implies that one corner
point of a weak one-sided GIC is determined exactly, it is achievable
by treating the interference as noise, and it is given by
\begin{align}
& \left( \frac{1}{2} \, \log(1+P_1), \, \frac{1}{2} \, \log\left(1 + \frac{P_2}{1+a P_1} \right) \right).
\label{eq: 1st exact corner point for Z-Gaussian IC}
\end{align}
In \cite[Theorem~2]{Sason04}, it is shown that this corner point achieves the sum-rate of the
weak one-sided GIC. We consider in the following
the second corner point of the capacity region:
according to Proposition~\ref{proposition: mixed interference}, the second
corner point of the weak one-sided GIC is given by $(R_1, C_2)$ where
(see \eqref{eq:maximal R1 for a two-user mixed Gaussian IC})
\begin{align}
\frac{1}{2} \, \log\left(1+\frac{aP_1}{1+P_2}\right)
\leq R_1 \leq \frac{1}{2} \, \log\left(1+\frac{P_1}{1+P_2}\right).
\label{eq: R1 for second corner point of Z-IC}
\end{align}
The lower bound on $R_1$ follows from the achievability of the point $(R_1^*, C_2)$
for the degraded GIC in \eqref{eq:Y1,2 for degraded Gaussian IC}.
The following statement summarizes this short discussion on weak one-sided GICs.
\begin{proposition}
{\em Consider a weak one-sided GIC, which in standard
form has power constraints $P_1$ and $P_2$ for transmitters~1 and~2, respectively, and
whose cross-link gains are $a_{12}=0$ and $a_{21}=a$ for $0<a<1$. One
of the two corner points of its capacity region is given in
\eqref{eq: 1st exact corner point for Z-Gaussian IC}, and it achieves the sum-rate.
The other corner point is $(R_1, C_2)$ where $R_1$ satisfies the bounds in
\eqref{eq: R1 for second corner point of Z-IC}, and these bounds are tight when
$a \rightarrow 1$.}
\label{proposition: corner points for a one-sided Gaussian IC with weak interference}
\end{proposition}

The achievable rate region of Costa \cite{Costa_ITA11} for a weak one-sided GIC
coincides with the Han-Kobayashi achievable region
for i.i.d. Gaussian codebooks (see \cite[Section~2]{Zhao_ISIT12}). This
region has a corner point at $(R_1^*, C_2)$ where $R_1^*$ is given in
\eqref{eq:maximal R1 in Costa's conjecture} with $a_{21}=a$ (note that it is
equal to the lower bound in \eqref{eq: R1 for second corner point of Z-IC}).
However, it remains unknown whether the capacity-achieving input distribution is Gaussian.

\subsection{Organization of this paper}
The structure of this paper is as follows:
Conjecture~\ref{conjecture: Costa's conjecture for the corner points} is considered in
Section~\ref{section: On the Corner Points of the Capacity Region of a Two-User Gaussian IC with Weak Interference}
for a weak GIC.
The excess rate for the sum-rate w.r.t. the corner points of the capacity region is considered in
Section~\ref{section: On the Sub-Optimality of the Corner Points for a Gaussian IC with Weak Interference}.
A summary is provided in Section~\ref{section: Summary and Outlook} with some directions for further research.
Throughout this paper, two-user GICs are considered.

\section{On the Corner Points of the Capacity Region of a Weak GIC}
\label{section: On the Corner Points of the Capacity Region of a Two-User Gaussian IC with Weak Interference}
This section considers Conjecture~\ref{conjecture: Costa's conjecture for the corner points}
for a weak GIC.
It is easy to verify that the points $(R_1, R_2) = (C_1, R_2^*)$ and $(R_1^*, C_2)$ are both included
in the capacity region of a weak GIC, and the corresponding receiver of the transmitter that operates
at the single-user capacity can be designed to decode the messages of the two users.
We proceed in the following to the converse part, which leads to the following statement:
\begin{theorem}
{\em Consider a weak two-user GIC, and
let $C_1$, $C_2$, $R_1^*$ and $R_2^*$ be as defined in
\eqref{eq:capacity of AWGNs}--\eqref{eq:maximal R2 in Costa's conjecture}.
If $R_1 \geq C_1 - \varepsilon$ for an arbitrary $\varepsilon>0$, then reliable communication
requires that
\begin{align}
R_2 \leq \min \biggl\{
& R_2^* + \frac{1}{2} \, \log\left(1 + \frac{P_2}{(1+a_{21} P_1)(1+a_{12} P_2)} \right)
+ 2 \varepsilon, \nonumber \\[0.2cm]
& \frac{1}{2} \, \log \left(1+\frac{P_2}{1+P_1}\right)
+ \Bigl(1+\frac{1+P_1}{a_{21} P_2} \Bigr) \varepsilon \biggr\}.
\label{eq:upper bound on R2 for weak interference}
\end{align}
Similarly, if $R_2 \geq C_2 - \varepsilon$, then
\begin{align}
R_1 \leq \min \biggl\{
& R_1^* + \frac{1}{2} \, \log\left(1 + \frac{P_1}{(1+a_{21} P_1)(1+a_{12} P_2)} \right)
+ 2 \varepsilon, \nonumber \\[0.2cm]
& \frac{1}{2} \, \log \left(1+\frac{P_1}{1+P_2}\right)
+ \Bigl(1+\frac{1+P_2}{a_{12} P_1} \Bigr) \varepsilon \biggr\}.
\label{eq:upper bound on R1 for weak interference}
\end{align}
Consequently, the corner points of the capacity region are $(R_1, C_2)$ and $(C_1, R_2)$ where
\begin{align}
& R_1^* \leq R_1 \leq \min \left\{ R_1^* + \frac{1}{2} \, \log\left(1 + \frac{P_1}{(1+a_{21} P_1)(1+a_{12} P_2)} \right), \; \frac{1}{2} \, \log \left(1+\frac{P_1}{1+P_2}\right) \right\}
\label{eq:bound1 on the corner points for weak interference} \\[0.2cm]
& R_2^* \leq R_2 \leq \min \left\{ R_2^* + \frac{1}{2} \, \log\left(1 + \frac{P_2}{(1+a_{21} P_1)(1+a_{12} P_2)} \right), \; \frac{1}{2} \, \log \left(1+\frac{P_2}{1+P_1}\right) \right\}.
\label{eq:bound2 on the corner points for weak interference}
\end{align}
In the limit where $P_1$ and $P_2$ tend to infinity, which makes it an interference-limited
channel,
\begin{enumerate}
\item Conjecture~\ref{conjecture: Costa's conjecture for the corner points} holds, and it gives an
asymptotically tight bound.
\item The rate pairs $(C_1, R_2^*)$ and $(R_1^*, C_2)$ form the corner points of the capacity region.
\item The answer to Question~\ref{question 2: corner points} is affirmative.
\end{enumerate}}
\label{theorem: bounds on the corner points of a GIC with weak interference}
\end{theorem}
\begin{proof}
The proof of this theorem relies on the two outer bounds on the capacity region
that are given in \cite[Theorem~3]{EtkinTseWang08} and \cite[Theorem~2]{Kramer04}.

Suppose that $R_1 \geq C_1-\varepsilon$ bits per channel use.
The outer bound by Etkin {\em et al.} in \cite[Theorem~3]{EtkinTseWang08} (it is also
known as the ETW bound) yields that the rates $R_1$ and $R_2$ satisfy the inequality
constraint
\begin{align*}
2 R_1 + R_2 \leq & \frac{1}{2} \, \log\bigl(1+P_1+a_{12} P_2\bigr) +
\frac{1}{2} \, \log\left(\frac{1+P_1}{1+a_{21}P_1} \right)
+ \frac{1}{2} \, \log\left(1+a_{21} P_1 + \frac{P_2}{1+a_{12} P_2} \right)
\end{align*}
which therefore yields that (see \eqref{eq:capacity of AWGNs} and \eqref{eq:maximal R2 in Costa's conjecture})
\begin{align}
R_2 \leq & \frac{1}{2} \, \log\bigl(1+P_1+a_{12} P_2\bigr) +
\frac{1}{2} \, \log\left(\frac{1+P_1}{1+a_{21}P_1} \right) + \frac{1}{2} \, \log\left(1+a_{21} P_1 + \frac{P_2}{1+a_{12} P_2} \right)
- \bigl( \log(1+P_1) - 2 \varepsilon \bigr) \nonumber \\[0.1cm]
= & R_2^* + \frac{1}{2} \, \log\left(1 + \frac{P_2}{(1+a_{21} P_1)(1+a_{12} P_2)} \right) + 2 \varepsilon.
\label{eq:first term in the bound on R_2}
\end{align}
The outer bound by Kramer in \cite[Theorem~2]{Kramer04}, formulated here in an equivalent form,
states that the capacity region is included
in the set $\mathcal{K} = \mathcal{K}_1 \cap \mathcal{K}_2$ where
\begin{align} \label{eq:outer-bound-Kramer}
   \mathcal{K}_1 & = \left\{ (R_1,R_2) :
   \begin{array}{l}
   0 \le R_1  \le \frac{1}{2} \, \log\left(1+\frac{(1-\beta)P'}{\beta P' + \frac{1}{a_{21}}} \right) \\[0.3cm]
   0 \le R_2  \le \frac{1}{2}\log(1 + \beta P')
   \end{array} \right\}
\end{align}
with $P' \triangleq P_2 + \frac{P_1}{a_{21}}$ and $\beta \in \bigl[\frac{P_2}{(1+P_1) P'}, \, \frac{P_2}{P'}\bigr]$
is a free parameter;
the set $\mathcal{K}_2$ is obtained by swapping the indices in $\mathcal{K}_1$. From the boundary of the outer bound in \eqref{eq:outer-bound-Kramer}, the value of $\beta$ that satisfies the equality
$$ \frac{1}{2} \, \log\left(1+\frac{(1-\beta)P'}{\beta P' + \frac{1}{a_{21}}} \right)
= C_1 - \varepsilon$$ is given by
$$\beta = \frac{2^{2 \varepsilon} P_2 + \frac{(2^{2 \varepsilon}-1) (1+P_1)}{a_{21}}}{(1+P_1) \, \left(P_2+\frac{P_1}{a_{21}}\right)} \, .$$
The substitution of this value of $\beta$ into the upper bound on $R_2$ in
\eqref{eq:outer-bound-Kramer} implies that if $R_1 \geq C_1 - \varepsilon$ then
\begin{align}
R_2 & \leq \frac{1}{2} \log(1 + \beta P') \nonumber \\[0.1cm]
& = \frac{1}{2} \, \log\left(1 + \frac{P_2}{1+P_1} \right) + \delta(\varepsilon)
\label{eq:1st step for the 2nd term of the bound on R_2}
\end{align}
where
$$\delta(\varepsilon) = \frac{1}{2} \log \left(1 + \frac{(2^{2\varepsilon}-1) \left(P_2+\frac{1+P_1}{a_{21}}\right)}{1+P_1+P_2} \right).$$
The function $\delta$ satisfies $\delta(0) = 0$, and straightforward calculus shows that
$$0 < \delta'(c) < 1+\frac{1+P_1}{a_{21} P_2}, \quad \forall \, c \geq 0.$$
It therefore follows (from the mean-value theorem of calculus) that
\begin{align}
0 < \delta(\varepsilon) < \left( 1+\frac{1+P_1}{a_{21} P_2} \right) \varepsilon.
\label{eq:2nd step for the 2nd term of the bound on R_2}
\end{align}
A combination of \eqref{eq:first term in the bound on R_2},
\eqref{eq:1st step for the 2nd term of the bound on R_2},
\eqref{eq:2nd step for the 2nd term of the bound on R_2} gives the
upper bound on the rate $R_2$ in \eqref{eq:upper bound on R2 for weak interference}.
Similarly, if $R_2 \geq C_2 - \varepsilon$, the upper bound on the rate $R_1$
in \eqref{eq:upper bound on R1 for weak interference} is obtained by swapping
the indices in \eqref{eq:upper bound on R2 for weak interference}.

From the inclusion of the points $(C_1, R_2^*)$ and $(R_1^*, C_2)$ in the capacity region,
and the bounds in \eqref{eq:upper bound on R2 for weak interference} and
\eqref{eq:upper bound on R1 for weak interference} in the limit where $\varepsilon \rightarrow 0$,
it follows that the corner points of the capacity region are $(R_1, C_2)$ and $(C_2, R_1)$ with
the bounds on $R_1$ and $R_2$ in \eqref{eq:bound1 on the corner points for weak interference} and
\eqref{eq:bound2 on the corner points for weak interference}, respectively.

Since the point $(C_1, R_2^*)$ is achievable, also is $(R_1, R_2^*)$ for $R_1 < C_1$; hence, if
$C_1 - \varepsilon \leq R_1 < C_1$, then the maximal rate $R_2$ of transmitter~2 satisfies
\begin{align*}
R_2^* \leq R_2 & \leq R_2^* + \frac{1}{2} \, \log\left(1 + \frac{P_2}{(1+a_{21} P_1)(1+a_{12} P_2)} \right)
+ 2 \varepsilon.
\end{align*}
The uncertainty in the maximal achievable rate $R_2$ when $R_1 \geq C_1 - \varepsilon$ and
$\varepsilon \rightarrow 0$ is therefore upper bounded by
$\Delta R_2 \triangleq \frac{1}{2} \, \log\left(1 + \frac{P_2}{(1+a_{21} P_1)(1+a_{12} P_2)} \right).$
The asymptotic case where $P_1, P_2 \rightarrow \infty$ and $\frac{P_2}{P_1} \rightarrow k$
for an arbitrary $k>0$ is examined in the following:
In this case, $R_2^* \rightarrow \frac{1}{2} \, \log(1+k a_{12})$ and $\Delta R_2 \rightarrow 0$
which proves that Conjecture~\ref{conjecture: Costa's conjecture for the corner points}
holds in this asymptotic case where the transmitted powers tend to infinity. Since the points
$(C_1, R_2^*)$ and $(R_1^*, C_2)$ are included in the capacity region, it follows from this
converse that they asymptotically form the corner points of this region.
As is explained above, operating at the points $(C_1, R_2^*)$ or $(R_1^*, C_2)$ enables
receiver~1 or~2, respectively, to decode both messages. This answers
Question~\ref{question 2: corner points} in the affirmative for the considered asymptotic case.
\end{proof}

\begin{remark}
{\em Consider a weak symmetric GIC where
$P_1 = P_2 = P$ and $a_{12} = a_{21} = a \in (0,1)$. The corner points of the
capacity region of this two-user interference channel are given by $(C, R_{\text{c}})$
and $(R_{\text{c}}, C)$ where $C = \frac{1}{2} \, \log(1+P)$ is the capacity
of a single-user AWGN channel with input power constraint $P$, and an
additive Gaussian noise with zero mean and unit variance.
Theorem~\ref{theorem: bounds on the corner points of a GIC with weak interference}
gives that
\begin{align}
R_{\text{c}} \leq  \min \left\{ \frac{1}{2} \, \log\left(1+\frac{aP}{1+P}\right)
+ \frac{1}{2} \, \log\left(1+\frac{P}{(1+aP)^2}\right), \;
\frac{1}{2} \, \log\left(1+\frac{P}{1+P}\right) \right\}.
\label{eq:upper bound on the corner point for symmetric GIC}
\end{align}
In the following, we compare the two terms inside the minimization in
\eqref{eq:upper bound on the corner point for symmetric GIC} where the first
term follows from the ETW bound in \cite[Theorem~3]{EtkinTseWang08},
and the second term follows from Kramer's bound in \cite[Theorem~2]{Kramer04}.
Straightforward algebra reveals that, for $a \in (0,1)$, the first
term gives a better bound on $R_{\text{c}}$ if and only if
\begin{align}
P > \frac{2 a^2 - a + 1 + \sqrt{5 a^2 - 2a + 1}}{2 a^2 (1-a)} \, .
\label{eq:condition on P where the ETW bound is tighter than Kramer's bound}
\end{align}
Hence, for an arbitrary cross-link gain $a \in (0,1)$ of a symmetric and weak two-user GIC,
there exists a threshold for the SNR where above it, the ETW bound provides a better upper
bound on the corner points; on the other hand, for values of SNR below
this threshold, Kramer's bound provides a better bound on the corner points. The
dependence of the threshold for the SNR $(P)$ on the cross-link gain is shown in
Figure~\ref{Figure:critical SNR}.
\begin{figure}[here!]
\begin{center}
\epsfig{file=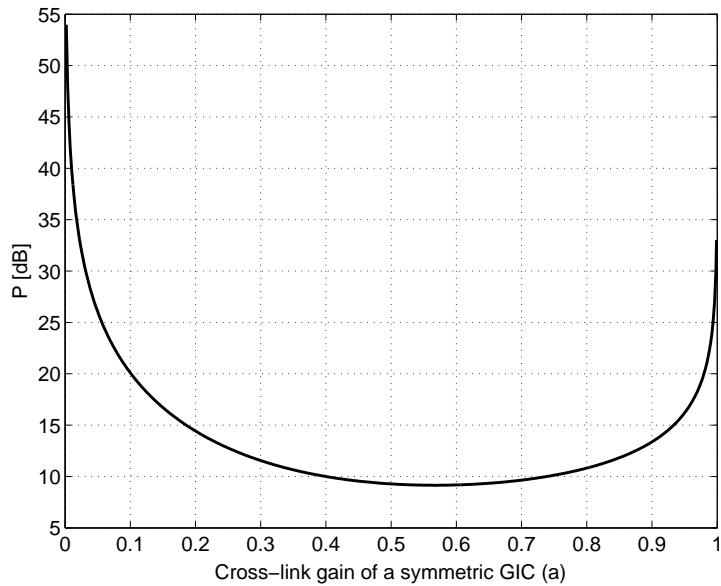,scale=0.55}
\caption{\label{Figure:critical SNR}
The curve in this figure shows the threshold for the SNR $(P)$, in decibels,
as a function of the cross-link gain $(a)$ for a weak and symmetric GIC.
This threshold is given by the right-hand side of
\eqref{eq:condition on P where the ETW bound is tighter than Kramer's bound}.
For points $(a, P)$ above this curve, the ETW bound is better in providing
an upper bound on the corner points of the capacity region, whereas Kramer's bound
is better in this respect for points $(a, P)$ below this curve.}
\end{center}
\end{figure}
The threshold for the SNR ($P$), as is shown in Figure~\ref{Figure:critical SNR},
tends to infinity if $a \rightarrow 0$ or $a \rightarrow 1$; this implies that in these two
cases, Kramer's bound is better for all values of $P$. This is further
discussed in the following:
\begin{enumerate}
\item If $a \rightarrow 0$ then, for every $P > 0$,
the first term on the right-hand side of
\eqref{eq:upper bound on the corner point for symmetric GIC}
tends to the capacity $C$; this forms a trivial upper bound
on the value $R_{\text{c}}$ of the corner point.
On the other hand, the second term on the right-hand side of
\eqref{eq:upper bound on the corner point for symmetric GIC}
gives the upper bound of $\frac{1}{2} \, \log\left(1+\frac{P}{1+P}\right)$
which is smaller than $C$ for all values of $P$. Note that the second term in
\eqref{eq:upper bound on the corner point for symmetric GIC} implies
that, for a symmetric GIC, $R_{\text{c}} \leq \frac{1}{2}$ bit per channel use
for all values of $P$. In fact, for a given $P$, the advantage of the second term
in the extreme case where $a \rightarrow 0$ served as the initial motivation for
incorporating it in
Theorem~\ref{theorem: bounds on the corner points of a GIC with weak interference}.
\item If $a \rightarrow 1$ then, for every $P > 0$, the first term tends to
$\frac{1}{2} \, \log\left(1+\frac{P}{1+P}\right)
+ \frac{1}{2} \, \log\left(1+\frac{P}{(1+P)^2}\right)$ which is larger than
the second term. Hence, also in this case, the second term gives a better bound
for all values of~$P$.
\end{enumerate}}
\end{remark}

\begin{example}
{\em The condition in \eqref{eq:condition on P where the ETW bound is tighter than Kramer's bound}
is consistent with \cite[Figs. 10 and 11]{Motahari09}, as explained in the following:
\begin{enumerate}
\item According to \cite[Fig. 10]{Motahari09}, for $P = 7$ and $a = 0.2$,
Kramer's outer bound gives a better upper bound on the corner point than
the ETW bound.
For $a = 0.2$, the complementary of the condition in
\eqref{eq:condition on P where the ETW bound is tighter than Kramer's bound}
implies that Kramer's bound is indeed better in this respect for $P <  27.725$.
This is supported by Figure~\ref{Figure:critical SNR}.
\item According to \cite[Fig. 11]{Motahari09}, for $P=100$ and $a = 0.1$, the
ETW is nearly as tight as Kramer's bound in providing an upper bound on the corner point.
For $a = 0.1$, the complementary of the condition in
\eqref{eq:condition on P where the ETW bound is tighter than Kramer's bound}
implies that Kramer's outer bound gives a better upper bound on the corner point than
the ETW bound if $P < 102.33$ (as is supported by Figure~\ref{Figure:critical SNR});
hence, for $P=100$, there is only a slight advantage
to Kramer's bound over the ETW bound that is not visible in \cite[Fig. 11]{Motahari09}:
Kramer's bound gives an upper bound on $R_{\text{c}}$ that is equal to 0.4964 bits per channel
use, and the ETW bound gives an upper bound of 0.5026 bits per channel use.
\end{enumerate}}
\end{example}

\vspace*{0.1cm}
\begin{remark}
{\em If $a_{12} a_{21} P_{1,2} \gg 1$, then it follows from
\eqref{eq:bound1 on the corner points for weak interference} and
\eqref{eq:bound2 on the corner points for weak interference}
that the two corner points of the capacity region approximately coincide with
the points $(R_1^*, C_2)$ and $(C_1, R_2^*)$ in
Conjecture~\ref{conjecture: Costa's conjecture for the corner points}.}
\label{remark: on the accuracy of the assessment of the corner points for weak interference}
\end{remark}

\vspace*{0.1cm}
In the following example, we evaluate the bounds in
Theorem~\ref{theorem: bounds on the corner points of a GIC with weak interference}
for finite values of transmitted powers ($P_1$ and $P_2$) to illustrate the asymptotic
tightness of these bounds.

\begin{example}
{\em Consider a weak and symmetric GIC where $a=0.5$ and $P=100$.
Assume that transmitter~1 operates at the single-user capacity
$C = \frac{1}{2} \log(1+P) = 3.33$~bits per channel use. According to
\eqref{eq:bound2 on the corner points for weak interference}, the corresponding
maximal rate $R_2$ of transmitter~2 is between 0.292 and 0.317 bits per channel
use; the upper bound on $R_2$ in this case follows from the ETW bound. This gives
good accuracy in the assessment of the two corner points of the capacity region (see
Remark~\ref{remark: on the accuracy of the assessment of the corner points for weak interference}
where, in this case, $a^2 P = 25 \gg 1$).
If $P$ is increased by 10~dB (to 1000), and transmitter~1 operates at the single-user capacity
$C = \frac{1}{2} \log(1+P) = 5.0$~bits per channel use, then the corresponding
maximal rate $R_2$ is between 0.292 and 0.295~bits per channel use.
Hence, the precision of the assessment of the corner points is improved in the latter case.
The improved accuracy of the latter assessment when the value of $P$ is increased is consistent with
Remark~\ref{remark: on the accuracy of the assessment of the corner points for weak interference},
and the asymptotic tightness of the bounds in
Theorem~\ref{theorem: bounds on the corner points of a GIC with weak interference}.
\begin{figure}[here!]
\begin{center}
\epsfig{file=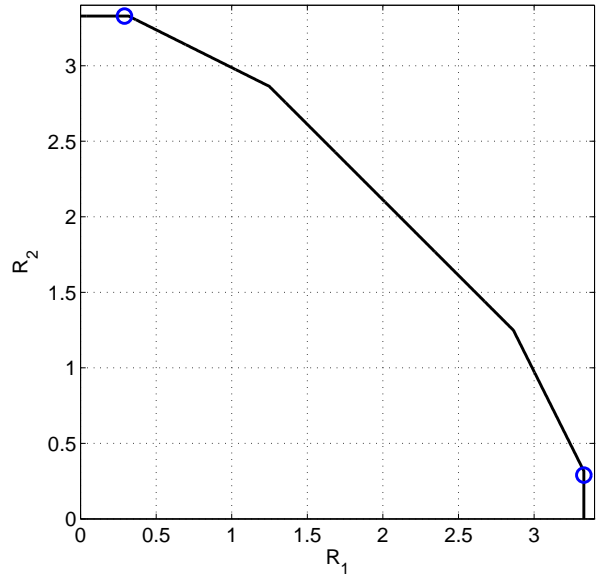,scale=0.55}
\caption{\label{Figure:on the corner points of the capacity region - weak interference}
The figure refers to a GIC
with cross-link gains $a_{12}=a_{21}=0.5$, and a common transmitted power
$P_1 = P_2 = 100$ in standard form (see
Example~\ref{example:corner points conjecture for Gaussian IC with weak interference}).
The solid curve is the boundary of the outer bound in
\cite[Theorem~3]{EtkinTseWang08} (the ETW bound in
\eqref{eq: ETW bound for GIC with weak interference}), and the two circled points refer to
Conjecture~\ref{conjecture: Costa's conjecture for the corner points};
these points are achievable, and they almost coincide with the boundary of the outer
bound.}
\end{center}
\end{figure}
Figure~\ref{Figure:on the corner points of the capacity region - weak interference}
refers to a weak and symmetric GIC where $P_1=P_2=100$ and $a_{12} = a_{21} = 0.5$.
The solid line in this figure corresponds to the boundary of the ETW outer bound on
the capacity region (see \cite[Theorem~3]{EtkinTseWang08}) which is given (in units
of bits per channel use) by
\begin{align} \label{eq: ETW bound for GIC with weak interference}
   \mathcal{R}_{\text o} & = \left\{ (R_1,R_2) :
   \begin{array}{l}
   0 \le R_1  \le 3.3291 \\[0.1cm]
   0 \le R_2  \le 3.3291 \\[0.1cm]
   R_1 + R_2  \le 4.1121 \\[0.1cm]
   2R_1 + R_2 \le 6.9755 \\[0.1cm]
   R_1 + 2R_2 \le 6.9755
   \end{array} \right\}.
\end{align}
The two circled points correspond to Conjecture~\ref{conjecture: Costa's conjecture for the corner points};
these points are achievable, and (as is verified numerically)
they almost coincide with the boundary of the outer
bound in \cite[Theorem~3]{EtkinTseWang08}.}
\label{example:corner points conjecture for Gaussian IC with weak interference}
\end{example}

\section{The Excess Rate for the Sum-Rate w.r.t. the Corner Points of the Capacity Region}
\label{section: On the Sub-Optimality of the Corner Points for a Gaussian IC with Weak Interference}
The sum-rate of a mixed, strong or one-sided GIC is attained at a corner point of its capacity
region. This is in contrast to a (two-sided) weak GIC whose sum-rate is not attained at a
corner point of its capacity region. It is therefore of interest to examine the excess rate
for the sum-rate w.r.t. these corner points by measuring the gap between the sum-rate
$(C_{\text{sum}})$ and the maximal total rate ($R_1+R_2$) at the corner points of the capacity region:
\begin{align}
& \Delta \triangleq C_{\text{sum}} - \max \bigl\{R_1+R_2 \colon (R_1, R_2) \; \text{is a corner point} \bigr\}.
\label{eq:Delta}
\end{align}
The parameter $\Delta$ measures the excess rate for the sum-rate w.r.t. the case where one transmitter
operates at its single-user capacity, and the other reduces its rate to the point where reliable
communication is achievable.
We have $\Delta=0$ for mixed, strong and one-sided GICs. This section derives bounds on $\Delta$ for
weak GICs, and it also provides an asymptotic analysis analogous
to the study of the generalized degrees of freedom (where the SNR and INR scalings are coupled such
that $\frac{\log(\text{INR})}{\log(\text{SNR})} = \alpha \geq 0$). This leads to an asymptotic
characterization of this gap which is demonstrated to be exact for the whole range of $\alpha$.
The upper and lower bounds on $\Delta$ are shown in this section to be asymptotically tight in
the sense that they achieve the exact asymptotic characterization. Improvements of the bounds
on $\Delta$ are derived in this section for finite SNR and INR, and these bounds are
exemplified numerically.

For the analysis in this section, the bounds in
Theorem~\ref{theorem: bounds on the corner points of a GIC with weak interference} and bounds on
the sum-rate (see, e.g., \cite{CostaNair-ITA12}, \cite{EtkinTseWang08}, \cite{Etkin_ISIT09}, \cite{Kramer09}, \cite{TuninettiW_ISIT2008} and \cite{Zhao_Suhas12}) are used to obtain upper and lower bounds
on $\Delta$.

\subsection{An Upper Bound on $\Delta$ for Weak GICs}
\label{subsection: An upper bound on delta for the corner points of the capacity region}
The following derivation of an upper bound on $\Delta$ relies on an upper bound on
the sum-rate, and a lower bound on the maximal value of $R_1 + R_2$ at the two corner points
of its capacity region. Since the points $(R_1^*, C_2)$ and $(C_1, R_2^*)$ are achievable
for a weak GIC, it follows that
\begin{align}
& \max \bigl\{R_1+R_2 \colon (R_1, R_2) \; \text{is a corner point} \bigr\} \nonumber \\
& \geq \max \{R_1^* + C_2, R_2^* + C_1 \} \nonumber \\
& = \frac{1}{2} \, \max \Bigl\{\log(1+P_2+a_{21} P_1), \log(1+P_1+a_{12} P_2) \Bigr\}.
\label{eq:lower bound on total throughput for the two corner points - weak interference}
\end{align}
An outer bound on the capacity region of a weak GIC is provided in
\cite[Theorem~3]{EtkinTseWang08}. This bound leads to the following
upper bound on the sum-rate:
\begin{align}
& C_{\text{sum}} \leq
\frac{1}{2} \, \min \Biggl\{ \log(1 + P_1) + \log\left( 1 + \frac{P_2}{1 + a_{21}P_1} \right),
\; \log(1 + P_2) + \log\left( 1 + \frac{P_1}{1 + a_{12}P_2} \right), \nonumber \\
& \qquad \qquad \qquad \; \; \log\left( 1 + a_{12}P_2 + \frac{P_1}{1 + a_{21}P_1} \right)
+ \log\left( 1 + a_{21}P_1 + \frac{P_2}{1 + a_{12}P_2} \right) \Biggr\}.
\label{eq:upper bound on the sum-rate - weak interference}
\end{align}
Consequently, combining
\eqref{eq:Delta}--\eqref{eq:upper bound on the sum-rate - weak interference}
gives the following upper bound on $\Delta$:
\begin{align}
& \Delta \leq \frac{1}{2} \, \Biggl[ \min \biggl\{ \log(1 + P_1) + \log\left( 1 + \frac{P_2}{1 + a_{21}P_1} \right),
\; \log(1 + P_2) + \log\left( 1 + \frac{P_1}{1 + a_{12}P_2} \right), \nonumber \\
& \qquad \qquad \qquad \; \; \log\left( 1 + a_{12}P_2 + \frac{P_1}{1 + a_{21}P_1} \right)
+ \log\left( 1 + a_{21}P_1 + \frac{P_2}{1 + a_{12}P_2} \right) \biggr\} \nonumber \\
& \qquad \quad - \max \Bigl\{\log(1+P_2+a_{21} P_1), \log(1+P_1+a_{12} P_2) \Bigr\} \Biggr].
\label{eq:upper bound on Delta - Gaussian IC with weak interference}
\end{align}
For a weak and symmetric GIC, where $P_1 = P_2 = P$ and $a_{12} = a_{21} = a$ ($0 < a < 1$),
\eqref{eq:upper bound on Delta - Gaussian IC with weak interference} is simplified to
\begin{align}
\Delta &= \Delta(P,a) \nonumber \\
& \leq \frac{1}{2} \, \Biggl[ \min \biggl\{ \log(1 + P) + \log\left( 1 + \frac{P}{1+aP} \right),
\; 2 \log\left( 1 + aP + \frac{P}{1 + aP} \right) \biggr\}
 - \log\bigl(1+(1+a)P\bigr) \Biggr] \nonumber \\[0.1cm]
&= \frac{1}{2} \, \min\left\{ \log\left(\frac{1+P}{1+aP}\right),
\; \log\left(1+\frac{P}{(1+aP)^2} + \frac{aP \bigl[P+(1+aP)^2\bigr]}{(1+aP)
\bigl(1+(a+1)P\bigr)} \right) \right\}.
\label{eq:upper bound on Delta - symmetric Gaussian IC with weak interference}
\end{align}
Hence, in the limit where we let $P$ tend to infinity,
\begin{align}
\lim_{P \rightarrow \infty} \Delta(P,a) \leq \frac{1}{2} \, \log\left(\frac{1}{a}\right),
\quad \forall \, a \in (0,1).
\label{eq:upper bound on Delta - large P, symmetric Gaussian IC with weak interference}
\end{align}
Note that, for $a=1$, the capacity region is the polyhedron that is obtained from the
intersection of the capacity regions of the two underlying Gaussian multiple-access channels.
This implies that $\Delta(P,1)=0$, so
the bound in \eqref{eq:upper bound on Delta - large P, symmetric Gaussian IC with weak interference}
is continuous from the left at $a=1$.

\subsection{A Lower Bound on $\Delta$ for Weak GICs}
\label{subsection: A lower bound on delta for the corner points of the capacity region}
The following derivation of a lower bound on $\Delta$ relies on a lower bound on
the sum-rate, and an upper bound on the maximal value of $R_1 + R_2$ at the two corner points
of the capacity region.
From Theorem~\ref{theorem: bounds on the corner points of a GIC with weak interference},
the maximal total rate at the corner points of the capacity region of a weak GIC is upper
bounded as follows:
\begin{align}
& \max \bigl\{R_1+R_2 \colon (R_1, R_2) \; \text{is a corner point} \bigr\} \nonumber \\
& \stackrel{\text{(a)}}{\leq} \min \Biggl\{ \max \Biggl\{R_1^* + C_2 + \frac{1}{2} \,
\log\left(1+\frac{P_1}{(1+a_{21}P_1)(1+a_{12}P_2)} \right),
\nonumber \\
& \qquad \qquad \qquad \quad R_2^* + C_1 + \frac{1}{2} \,
\log\left(1+\frac{P_2}{(1+a_{21}P_1)(1+a_{12}P_2)}\right) \Biggr\}, \;
\frac{1}{2} \, \log(1+P_1+P_2) \Biggr\} \nonumber \\
& \stackrel{\text{(b)}}{=} \frac{1}{2} \, \min \Biggl\{ \max \Biggl\{\log(1+P_2+a_{21} P_1)
+ \log\left(1+\frac{P_1}{(1+a_{21}P_1)(1+a_{12}P_2)} \right), \nonumber \\
& \qquad \qquad \qquad \quad \; \log(1+P_1+a_{12} P_2) +
\log\left(1+\frac{P_2}{(1+a_{21}P_1)(1+a_{12}P_2)} \right) \Biggr\}, \;
\log(1+P_1+P_2) \Biggr\}
\label{eq:upper bound on total throughput for the two corner points - weak interference}
\end{align}
where inequality~(a) follows from \eqref{eq:bound1 on the corner points for weak interference},
\eqref{eq:bound2 on the corner points for weak interference}, and the equality
$$\max \bigl\{\min\{a,c\}, \, \min\{b,c\} \bigr\} = \min\bigl\{ \max \{a, b\}, c\bigr\}, \quad
\forall \, a, b, c \in \reals$$
and equality~(b) follows from \eqref{eq:capacity of AWGNs}--\eqref{eq:maximal R2 in Costa's conjecture}.

In order to get a lower bound on the sum-rate of the capacity region of a weak GIC,
we rely on the particularization of the outer bound in
\cite{TelatarTse07} for a GIC. This leads to the following outer bound $\mathcal{R}_{\text{o}}$ in \cite[Section~6.7.2]{ElGamalKim_book}:
\begin{align} \label{eq:outer-awgn}
   \mathcal{R}_{\text o} & = \left\{ (R_1,R_2) :
   \begin{array}{l}
   0 \le R_1  \le \frac{1}{2}\log(1 + P_1) \\[0.1cm]
   0 \le R_2  \le \frac{1}{2}\log(1 + P_2) \\[0.1cm]
   R_1 + R_2  \le \frac{1}{2}\log(1 + P_1 + a_{12}P_2)  +  \frac{1}{2}\log\left( 1 + \frac{P_2}{1 + a_{12}P_2} \right) \\[0.2cm]
   R_1 + R_2  \le \frac{1}{2}\log(1 + P_2 + a_{21}P_1)  +  \frac{1}{2}\log\left( 1 + \frac{P_1}{1 + a_{21}P_1} \right) \\[0.2cm]
   R_1 + R_2  \le \frac{1}{2}\log\left(1 + a_{12}P_2 + \frac{P_1}{1 + a_{21}P_1} \right)
   + \frac{1}{2}\log\left( 1 + a_{21}P_1 + \frac{P_2}{1 + a_{12}P_2} \right) \\[0.2cm]
   2R_1 + R_2 \le \frac{1}{2}\log(1 + P_1 + a_{12}P_2) + \frac{1}{2}\log\left(1 + \frac{P_1}{1 + a_{21}P_1} \right)\\[0.2cm]
   \qquad \qquad \quad + \frac{1}{2}\log\left(1 + a_{21}P_1 + \frac{P_2}{1 + a_{12}P_2} \right) \\[0.2cm]
   R_1 + 2R_2 \le \frac{1}{2}\log(1 + P_2 + a_{21}P_1) + \frac{1}{2}\log\left(1 + \frac{P_2}{1 + a_{12}P_2} \right)\\[0.2cm]
   \qquad \qquad \quad + \frac{1}{2}\log\left(1 + a_{12}P_2 + \frac{P_1}{1 + a_{21}P_1} \right)
   \end{array} \right\}.
\end{align}
The outer bound $\mathcal{R}_{\text o}$ has the property that if $(R_1, R_2) \in \mathcal{R}_{\text{o}}$ then
$(R_1 - \frac{1}{2}, R_2-\frac{1}{2}) \in \mathcal{R}_{\text{HK}}$
where $\mathcal{R}_{\text{HK}}$ denotes the Han-Kobayashi achievable rate region in \cite{Han81}
(see \cite[Remark~2]{TelatarTse07} and \cite[Section~6.7.2]{ElGamalKim_book}). Note that the "within
one bit" result in \cite{EtkinTseWang08} and \cite{TelatarTse07} is per complex dimension,
and it is replaced here by half a bit per dimension since all the random variables involved in the
calculations of the outer bound on the capacity region of a scalar GIC are real-valued
\cite[Theorem~6.6]{ElGamalKim_book}.
Consider the boundary of the outer bound $\mathcal{R}_{\text{o}}$ in \eqref{eq:outer-awgn}.
If one of the three inequality constraints on $R_1 + R_2$ is active in \eqref{eq:outer-awgn}
(this condition is first needed to be verified), then a point on the boundary of the rate region $\mathcal{R}_{\text{o}}$ that is dominated by one of these three inequality constraints satisfies the equality
\begin{align}
R_1 + R_2 & = \frac{1}{2} \, \min \Biggl\{\log(1 + P_1 + a_{12}P_2)  +  \log\left( 1 + \frac{P_2}{1 + a_{12}P_2} \right), \nonumber \\
& \qquad \qquad \quad \log(1 + P_2 + a_{21}P_1)  + \log\left( 1 + \frac{P_1}{1 + a_{21}P_1} \right), \nonumber \\
& \qquad \qquad \quad \log\left(1 + a_{12}P_2 + \frac{P_1}{1 + a_{21}P_1} \right)
+ \log\left( 1 + a_{21}P_1 + \frac{P_2}{1 + a_{12}P_2}\right) \Biggr\}.
\label{eq:maximal throughput that refers to the outer bound in Theorem 1 under a mild condition}
\end{align}
Since $(R_1 - \frac{1}{2}, R_2 - \frac{1}{2})$ is an achievable rate pair, then the sum-rate
is lower bounded by $R_1 + R_2 - 1$. It therefore follows from
\eqref{eq:maximal throughput that refers to the outer bound in Theorem 1 under a mild condition} that
\begin{align}
C_{\text{sum}} & \geq \frac{1}{2} \, \min \Biggl\{\log(1 + P_1 + a_{12}P_2)  +  \log\left( 1 + \frac{P_2}{1 + a_{12}P_2} \right), \nonumber \\[0.2cm]
& \qquad \qquad \quad \log(1 + P_2 + a_{21}P_1)  + \log\left( 1 + \frac{P_1}{1 + a_{21}P_1} \right), \nonumber \\
& \qquad \qquad \quad \log\left(1 + a_{12}P_2 + \frac{P_1}{1 + a_{21}P_1} \right)
+ \log\left( 1 + a_{21}P_1 + \frac{P_2}{1 + a_{12}P_2}\right) \Biggr\} - 1.
\label{eq:lower bound on the sum-rate - weak interference}
\end{align}
A combination of \eqref{eq:Delta},
\eqref{eq:upper bound on total throughput for the two corner points - weak interference} and
\eqref{eq:lower bound on the sum-rate - weak interference} leads to the following lower bound
on the excess rate for the sum-rate w.r.t. the corner points:
\begin{align}
\hspace*{-0.25cm} \Delta & \geq \frac{1}{2} \, \Biggl[ \min \Biggl\{\log(1 + P_1 + a_{12}P_2)  +  \log\left( 1 + \frac{P_2}{1 + a_{12}P_2} \right), \nonumber \\[0.2cm]
& \qquad \qquad \quad \log(1 + P_2 + a_{21}P_1)  + \log\left( 1 + \frac{P_1}{1 + a_{21}P_1} \right), \nonumber \\
& \qquad \qquad \quad \log\left(1 + a_{12}P_2 + \frac{P_1}{1 + a_{21}P_1} \right)
+ \log\left( 1 + a_{21}P_1 + \frac{P_2}{1 + a_{12}P_2}\right) \Biggr\} \nonumber \\
& \qquad - \min \Biggl\{ \max \Biggl\{\log(1+P_2+a_{21} P_1) + \log\left(1+\frac{P_1}{(1+a_{21}P_1)(1+a_{12}P_2)} \right),  \nonumber \\
& \qquad \qquad \qquad \qquad \log(1+P_1+a_{12} P_2) + \log\left(1+\frac{P_2}{(1+a_{21}P_1)(1+a_{12}P_2)} \right) \Biggr\}, \log(1+P_1+P_2) \Biggr\} \Biggr] - 1
\label{eq:lower bound on Delta - Gaussian IC with weak interference}
\end{align}
provided that there exists a rate-pair $(R_1, R_2)$ that is dominated by one of the
three inequality constraints on $R_1+R_2$ in \eqref{eq:outer-awgn}; as mentioned above,
this condition is first needed to be verified for validating both lower bounds
in \eqref{eq:lower bound on the sum-rate - weak interference} and
\eqref{eq:lower bound on Delta - Gaussian IC with weak interference}.
In the following, the lower bound on $\Delta$ is particularized
for a weak and symmetric GIC, and a sufficient condition is stated
for ensuring that the lower bounds in
\eqref{eq:lower bound on the sum-rate - weak interference} and
\eqref{eq:lower bound on Delta - Gaussian IC with weak interference}
hold for this channel. To this end, we state and prove the following lemma:
\begin{lemma}
{\em For a weak and symmetric two-user GIC with a common power constraint
on its inputs that satisfies $P \geq 2.551$, there exists a rate-pair
$(R_1, R_2)$ on the boundary of the outer bound $\mathcal{R}_{\text{o}}$ in
\eqref{eq:outer-awgn} that is dominated by one of the inequality constraints
on $R_1+R_2$ in $\mathcal{R}_{\text{o}}$.}
\label{lemma for a symmetric Gaussian IC}
\end{lemma}
\begin{proof}
Consider the straight lines that correspond to the inequality constraints on
$2R_1+R_2$ and $R_1+2R_2$ in \eqref{eq:outer-awgn}. For a weak and symmetric
two-user GIC (where $a_{12}=a_{21}=a$ with $0<a<1$, and $P_1 = P_2 \triangleq P$),
this corresponds to
\begin{align*}
& 2 R_1 + R_2 = \frac{1}{2} \left[ \log(1+P+aP) + \log\left(1+\frac{P}{1+aP}\right)
+ \log\left(1+aP+\frac{P}{1+aP}\right) \right], \nonumber \\[0.1cm]
& R_1 + 2 R_2 = \frac{1}{2} \left[ \log(1+P+aP) + \log\left(1+\frac{P}{1+aP}\right)
+ \log\left(1+aP+\frac{P}{1+aP}\right) \right].
\end{align*}
These two straight lines intersect at a point $(R_1, R_2)$ where
\begin{align}
R_1=R_2 \triangleq R = \frac{1}{6} \left[ \log(1+P+aP) + \log\left(1+\frac{P}{1+aP}\right)
+ \log\left(1+aP+\frac{P}{1+aP}\right) \right]
\label{eq:R}
\end{align}
and the corresponding value of $R_1+R_2$ at this point is given by
\begin{align}
& R_1 + R_2 = \frac{1}{3} \left[ \log(1+P+aP) + \log\left(1+\frac{P}{1+aP}\right)
+ \log\left(1+aP+\frac{P}{1+aP}\right) \right].
\label{eq:value of R-sum at the intersecting point}
\end{align}
For the considered GIC, the inequality constraints on $R_1+R_2$ in the outer bound
\eqref{eq:outer-awgn} are given by
\begin{align}
& R_1 + R_2 \leq \frac{1}{2} \, \left[\log(1+P+aP) + \log\left(1+\frac{P}{1+aP}\right) \right],
\label{eq:1st inequality for R1+R2} \\
& R_1 + R_2 \leq \log\left(1+aP+\frac{P}{1+aP}\right).
\label{eq:2nd inequality for R1+R2}
\end{align}
The right-hand side of \eqref{eq:value of R-sum at the intersecting point}
is equal to the weighted average of the right-hand sides of
\eqref{eq:1st inequality for R1+R2} and \eqref{eq:2nd inequality for R1+R2}
with weights $\frac{2}{3}$ and $\frac{1}{3}$, respectively. Hence, it follows
that one of the two inequality constraints on $R_1+R_2$ in
\eqref{eq:1st inequality for R1+R2} and \eqref{eq:2nd inequality for R1+R2} should be
active in the determination of the boundary of the outer bound in \eqref{eq:outer-awgn},
provided that the point $(R,R)$ satisfies the condition $R<\frac{1}{2} \, \log(1+P)$
for every $0 < a < 1$ (see the first and second inequality constraints on $R_1$ and
$R_2$, respectively, in \eqref{eq:outer-awgn}). By showing this, it implies
that the point $(R,R)$ is outside the rate region $\mathcal{R}_{\text{o}}$ in
\eqref{eq:outer-awgn}. Consequently, it ensures the existence of a point $(R_1, R_2)$,
located at the boundary of the rate region in \eqref{eq:outer-awgn}, that is
dominated by one of the inequality constraints on $R_1+R_2$ in
\eqref{eq:1st inequality for R1+R2} and \eqref{eq:2nd inequality for R1+R2}.
In order to verify that indeed the condition $R<\frac{1}{2} \, \log(1+P)$ holds
for every $0 < a < 1$, where $R$ is given in \eqref{eq:R}, let
\begin{align}
f_P(a) \triangleq \frac{1}{2} \, \log(1+P) - \frac{1}{6} \left[ \log(1+P+aP) + \log\left(1+\frac{P}{1+aP}\right)
+ \log\left(1+aP+\frac{P}{1+aP}\right) \right], \quad \forall \, a \in [0,1]
\label{eq: function f_P}
\end{align}
where $P>0$ is arbitrary; the satisfiability of this condition requires that
$f_P$ is positive over the interval [0,1].
The function $f_P$ satisfies $f_P(0)=0$, and it is concave over the interval
[0,1] if and only if $P \geq 0.680$ (one can show that this is a necessary and
sufficient condition such that $f_P''(1) \leq 0$; furthermore, under the latter
condition, the third derivative of $f_P$ is also positive on [0,1], which implies that
$f_P'' \leq 0$ over this interval, so the function $f_P$ is concave).
This implies that $f_P(a) > 0$
for all $a \in (0,1)$ if $f_P(1) \geq 0$ and $P \geq 0.680$.
Straightforward algebra shows that $f_P(1) \geq 0$ if and only if
$P^4 + P^3 - 6P^2 - 7P - 2 \geq 0$, which is satisfied if and only
if $P \geq 2.55003$ (the other solutions of this inequality are
infeasible for $P$ since it is real and positive).
This completes the proof of the lemma.
\end{proof}

Lemma~\ref{lemma for a symmetric Gaussian IC} yields that the lower bound on the sum-rate in
\eqref{eq:lower bound on the sum-rate - weak interference} is satisfied for a weak and symmetric
GIC if $P \geq 2.551$. Consequently, also the lower bound on $\Delta$ in
\eqref{eq:lower bound on Delta - Gaussian IC with weak interference}
holds for a weak and symmetric GIC under the same condition on $P$. In this case, the lower bound in
\eqref{eq:lower bound on Delta - Gaussian IC with weak interference} is simplified to
\begin{align}
\Delta &= \Delta(P,a) \nonumber \\
& \geq \frac{1}{2} \, \Biggl[ \min \Biggl\{\log\bigl(1 + (a+1)P\bigr) +  \log\left( 1 + \frac{P}{1 + aP} \right),
\, 2 \log\left(1 + aP + \frac{P}{1 + aP} \right) \Biggr\} \nonumber \\
& \qquad - \min \left\{ \log\bigl(1+(a+1)P\bigr) + \log\left(1+\frac{P}{(1+aP)^2} \right),
\log(1+2P) \right\} \Biggr] - 1.
\label{eq:lower bound on Delta - symmetric Gaussian IC with weak interference}
\end{align}
In the following, we consider the limit of the lower bound on $\Delta$ in the asymptotic case
where we let $P$ tend to infinity, while $a \in (0,1)$ is kept fixed. In this case, we have from
the lower bound in \eqref{eq:lower bound on Delta - symmetric Gaussian IC with weak interference}
\begin{align}
& \lim_{P \rightarrow \infty} \Delta(P,a) \nonumber \\
& \geq \frac{1}{2} \, \lim_{P \rightarrow \infty} \Biggl[ \min \Biggl\{\log\bigl(1 + (a+1)P\bigr) +  \log\left( 1 + \frac{P}{1 + aP} \right),
\, \log \biggl( \left(1 + aP + \frac{P}{1 + aP} \right)^2 \biggr) \Biggr\} \nonumber \\
& \quad \qquad \qquad - \min \left\{ \log\bigl(1+(a+1)P\bigr) + \log\left(1+\frac{P}{(1+aP)^2} \right),
\log(1+2P) \right\} \Biggr] - 1 \nonumber \\
& \stackrel{(\text{a})}{=} \frac{1}{2} \, \lim_{P \rightarrow \infty} \Biggl[ \log\bigl(1 + (a+1)P\bigr)
+  \log\left( 1 + \frac{P}{1 + aP} \right)
- \left( \log\bigl(1+(a+1)P\bigr) + \log\left(1+\frac{P}{(1+aP)^2}\right) \right) \Biggr] - 1 \nonumber \\
& = \frac{1}{2} \, \lim_{P \rightarrow \infty} \left[ \log\left( 1 + \frac{P}{1 + aP} \right)
- \log\left(1+\frac{P}{(1+aP)^2}\right) \right] - 1 \nonumber \\[0.1cm]
& = \frac{1}{2} \, \log\left(1+\frac{1}{a}\right) - 1.
\label{eq:lower bound on Delta - large P, symmetric Gaussian IC with weak interference}
\end{align}
Equality~(a) holds since, for large enough $P$,
\begin{align*}
& \log\bigl(1+(a+1)P\bigr) + \log\left(1+\frac{P}{1+aP}\right)
\approx \log\bigl((a+1)P\bigr) + \log\left(1+\frac{1}{a}\right)
= \log \left(\frac{(a+1)^2 P}{a} \right), \\
& 2 \log\left(1 + aP + \frac{P}{1+aP}\right) \approx \log(a^2 P^2), \quad
\log\bigl(1+(a+1)P\bigr) + \log\left(1+\frac{P}{(1+aP)^2} \right) \approx \log\bigl((a+1)P\bigr)
\end{align*}
so, if $a \in (0,1)$ and $P$ is large enough, each minimization of the pair of terms
in the two lines before equality~(a) is equal to its first term.

For a weak and symmetric GIC,
a comparison of the asymptotic upper and lower bounds on $\Delta$ in
\eqref{eq:upper bound on Delta - large P, symmetric Gaussian IC with weak interference}
and \eqref{eq:lower bound on Delta - large P, symmetric Gaussian IC with weak interference}
yields that these two asymptotic bounds differ by at most~1 bit per channel use; this
holds irrespectively of the cross-link gain $a \in (0,1)$.
Note that the upper bound is tight for $a$ close to~1, and also both asymptotic bounds
scale like $\frac{1}{2} \, \log\left(\frac{1}{a}\right)$ for small values of $a$ (so, they
tend to infinity as $a \rightarrow 0$).

\subsection{An Analogous Measure to the Generalized Degrees of Freedom and its Implications}
\label{subsection:An analogous measure to the generalized degrees of freedom}
This section is focused on the model of a two-user symmetric GIC, and it provides
an asymptotic analysis of the excess rate for the sum-rate w.r.t. the corner points of its
capacity region. The asymptotic analysis of this excess rate $(\Delta)$ is analogous to
the study of the generalized degrees of freedom where the SNR and INR
scalings are coupled such that
\begin{align}
\frac{\log(\text{INR})}{\log(\text{SNR})} = \alpha \geq 0.
\label{eq:scaling of SNR and INR}
\end{align}
The main results of this section is a derivation of an exact asymptotic characterization of
$\Delta$ for the whole range of $\alpha$ (see Theorem~\ref{theorem:closed-form expression for delta}),
and a demonstration that the closed-form expressions for the upper and lower bounds on
$\Delta$ in Sections~\ref{subsection: An upper bound on delta for the corner points of the capacity region}
and~\ref{subsection: A lower bound on delta for the corner points of the capacity region}
are asymptotically tight in the sense of achieving the exact asymptotic characterization of
$\Delta$ (see Theorem~\ref{theorem: asymptotic tightness of the bounds on the excess-rate}).
Implications of the asymptotic analysis and the main results of this section are further
discussed in the following.

Consider a two-user symmetric GIC whose cross-link gain $a$ scales like
$P^{\alpha-1}$ for some fixed value of $\alpha \geq 0$. For this GIC, the generalized
degrees of freedom (GDOF) is defined as the asymptotic limit of the normalized sum-rate $\frac{C_{\text{sum}}(P,P^{\alpha-1})}{\log P}$ when $P \rightarrow \infty$.
This GDOF refers to the case where the SNR $(P)$ tends to infinity, and
the interference to noise ratio $(\text{INR} = aP)$ scales according to
\eqref{eq:scaling of SNR and INR} while $\alpha \geq 0$ is kept fixed.
The GDOF of a two-user symmetric GIC (without feedback) is defined as follows:
\begin{align}
d(\alpha) &\triangleq \lim_{P \rightarrow \infty} \frac{C_{\text{sum}}(P,P^{\alpha-1})}{\log P}
\label{eq:definition of GDOF}
\end{align}
and this limit exists for every $\alpha \geq 0$ (see \cite[Section~3.G]{EtkinTseWang08}).

For large $P$, let us consider in an analogous way the asymptotic scaling of the
normalized excess rate for the sum-rate w.r.t. the corner points of the capacity region.
To this end, we study the asymptotic limit of the ratio $\frac{\Delta(P,P^{\alpha-1})}{\log P}$
for a fixed $\alpha \geq 0$ when $P$ tends to infinity. Similarly to \eqref{eq:definition of GDOF},
the denominator of this ratio is equal to the asymptotic sum-rate of two parallel AWGN channels with no
interference.
However, in the latter expression, the excess rate for the sum-rate w.r.t. the corner points is replacing
the sum-rate that appears in the numerator on the right-hand side of \eqref{eq:definition of GDOF}.
Correspondingly, for an arbitrary $\alpha \geq 0$, let us define
\begin{align}
\delta(\alpha) \triangleq \lim_{P \rightarrow \infty} \frac{\Delta(P, P^{\alpha-1})}{\log P}.
\label{eq:delta - analogous definition to GDOF}
\end{align}
provided that this limit exists. In the following, we demonstrate the existence of this limit
and provide a closed-form expression for $\delta$.

\begin{theorem}
{\em The limit in \eqref{eq:delta - analogous definition to GDOF} exists
for every $\alpha \geq 0$, and the function $\delta$ admits the following
closed-form expression:
\begin{align}
\delta(\alpha)
  &= \left\{ \begin{array}{cl}
  \left| \frac{1}{2}-\alpha \right|, & \mbox{if $0 \leq \alpha < \frac{2}{3}$} \\[0.1cm]
  \frac{1-\alpha}{2}, & \mbox{if $\frac{2}{3} \leq \alpha < 1$} \\[0.1cm]
  0, & \mbox{if $\alpha \geq 1$}
  \end{array} \right. .   \label{eq:closed-form expression for delta}
\end{align}}
\label{theorem:closed-form expression for delta}
\end{theorem}
\begin{proof}
If $\alpha \geq 1$ and $P \geq 1$, the cross-link gain is
$a = P^{\alpha-1} \geq 1$, and the channel is a strong and symmetric
two-user GIC.
The capacity region of a strong two-user GIC is equal to the
intersection of the capacity regions of the two Gaussian
multiple-access channels from the two transmitters to each one of
the receivers (see \cite[Theorem~5.2]{Han81} and \cite{Sato81}).
The sum-rate of this GIC is therefore equal to the total rate $(R_1+R_2)$
at each of the corner points of its capacity region. Hence, if
$\alpha \geq 1$ and $P > 1$ then $\Delta(P, P^{\alpha-1}) = 0$,
and \eqref{eq:delta - analogous definition to GDOF} implies that
\begin{align}
\delta(\alpha) = 0, \quad \forall \, \alpha \geq 1.
\label{eq:delta is zero}
\end{align}

For a symmetric two-user GIC with an input power constraint $P > 1$ and
an interference level $\alpha \in [0,1)$, the cross-link gain is
$a = P^{\alpha-1} < 1$. This refers to a weak and symmetric two-user GIC.
From Theorem~\ref{theorem: bounds on the corner points of a GIC with weak interference}
(see \eqref{eq:bound1 on the corner points for weak interference} and
\eqref{eq:bound2 on the corner points for weak interference}),
the bounds on the corner points of the capacity region of a weak and symmetric two-user GIC
imply that the maximal total rate at these corner points satisfies the inequality
\begin{align}
\frac{1}{2} \, \log(1+P) \leq
\max \bigl\{R_1+R_2 \colon (R_1, R_2) \; \text{is a corner point} \bigr\}
\leq \frac{1}{2} \, \log(1+2P).
\label{eq:bounds on the total rate at the corner points of a symmetric GIC}
\end{align}
From \eqref{eq:Delta} and
\eqref{eq:bounds on the total rate at the corner points of a symmetric GIC}, it follows
that for $P > 1$ and $\alpha \in [0,1)$
\begin{align}
C_{\text{sum}}(P,P^{\alpha-1}) - \frac{1}{2} \, \log(1+2P) \leq \Delta(P, P^{\alpha-1})
\leq C_{\text{sum}}(P,P^{\alpha-1}) - \frac{1}{2} \, \log(1+P).
\label{eq:bounds on the excess rate for a symmetric GIC}
\end{align}
Consequently, for $\alpha \in (0,1)$, a division by $\log P$ of the three sides
of the inequality in \eqref{eq:bounds on the excess rate for a symmetric GIC}
and a calculation of the limit as $P \rightarrow \infty$ gives that (see
\eqref{eq:definition of GDOF} and \eqref{eq:delta - analogous definition to GDOF})
\begin{align}
\delta(\alpha) = d(\alpha) - \frac{1}{2}, \quad \forall \, \alpha \in (0, 1).
\label{eq:connection between delta and GDOF for weak interference}
\end{align}
The limit in \eqref{eq:definition of GDOF} for the GDOF of a two-user symmetric GIC
(without feedback) exists, and it admits the following closed-form expression
(see \cite[Theorem~2]{EtkinTseWang08}):
\begin{align}
d(\alpha)
&= \min \Bigl\{ 1, \max\left\{\frac{\alpha}{2}, 1-\frac{\alpha}{2}\right\},
\max\left\{\alpha,1-\alpha\right\} \Bigr\} \nonumber \\[0.1cm]
  &= \left\{ \begin{array}{cl}
  1-\alpha, & \mbox{if $0 \leq \alpha < \frac{1}{2}$} \\[0.1cm]
  \alpha, & \mbox{if $\frac{1}{2} \leq \alpha < \frac{2}{3}$} \\[0.1cm]
  1-\frac{\alpha}{2}, & \mbox{if $\frac{2}{3} \leq \alpha < 1$} \\[0.1cm]
  \frac{\alpha}{2}, & \mbox{if $1 \leq \alpha < 2$} \\[0.1cm]
  1, & \mbox{if $\alpha \geq 2$}
  \end{array} \right. .   \label{eq:closed-form expression for the GDOF}
\end{align}
A combination of \eqref{eq:delta is zero},
\eqref{eq:connection between delta and GDOF for weak interference}
and \eqref{eq:closed-form expression for the GDOF} proves the closed-form
expression for $\delta$ in \eqref{eq:closed-form expression for delta}.
\end{proof}

\begin{figure}[here!]
\begin{center}
\epsfig{file=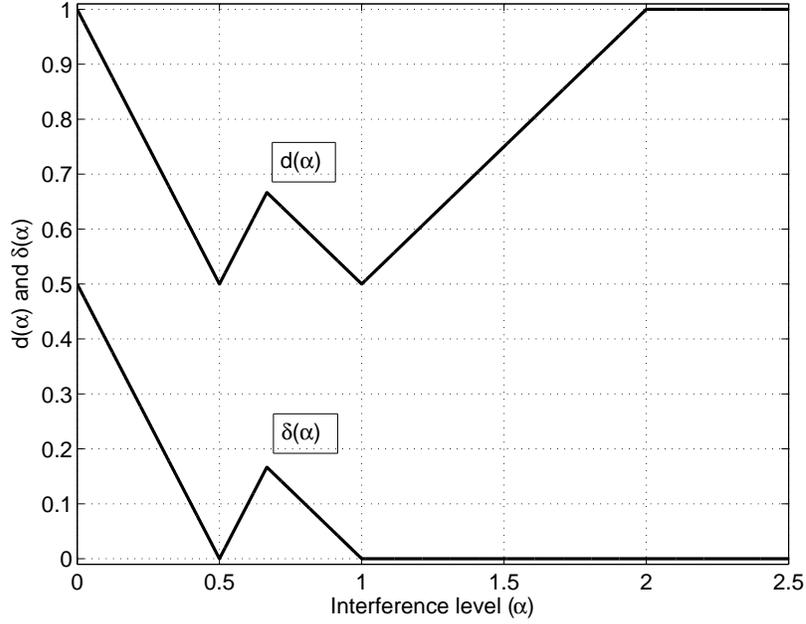,scale=0.6}
\caption{\label{Figure:GDOF_and_delta} A comparison of the generalized degrees of
freedom (GDOF) in \eqref{eq:closed-form expression for the GDOF}, and the exact asymptotic
characterization of $\delta$ in \eqref{eq:delta - analogous definition to GDOF} (see
\eqref{eq:closed-form expression for delta}) as a function of the non-negative interference
level $\alpha$.}
\end{center}
\end{figure}

Figure~\ref{Figure:GDOF_and_delta} provides a comparison of the GDOF and the function $\delta$
for an interference level $\alpha$ (i.e., the cross-link gain is $a=P^{\alpha-1}$).
Equation \eqref{eq:connection between delta and GDOF for weak interference} shows that, for
an interference level $\alpha \in [0,1]$, the difference between the GDOF (denoted here by
$d(\alpha)$) and $\delta(\alpha)$ is half a bit per channel use (see Figure~\ref{Figure:GDOF_and_delta}).
In light of the closed-form expression of $\delta$ in \eqref{eq:closed-form expression for delta},
the asymptotic tightness of the bounds in
\eqref{eq:upper bound on Delta - symmetric Gaussian IC with weak interference} and
\eqref{eq:lower bound on Delta - symmetric Gaussian IC with weak interference} is
demonstrated in the following:

\begin{theorem}
{\em Consider a weak and symmetric two-user GIC where the SNR and INR scalings are coupled according to
\eqref{eq:scaling of SNR and INR}.
Then, the upper and lower bounds on the excess rate $\Delta$
in \eqref{eq:upper bound on Delta - symmetric Gaussian IC with weak interference} and
\eqref{eq:lower bound on Delta - symmetric Gaussian IC with weak interference}
are asymptotically tight in the sense that, in the limit where $P \rightarrow \infty$,
the normalization of these bounds by $\log P$ tend to $\delta$ in
\eqref{eq:delta - analogous definition to GDOF} and
\eqref{eq:closed-form expression for delta}.}
\label{theorem: asymptotic tightness of the bounds on the excess-rate}
\end{theorem}
\begin{proof}
Substituting $a=P^{\alpha-1}$ into the upper bound on $\Delta(P,a)$ in
\eqref{eq:upper bound on Delta - symmetric Gaussian IC with weak interference}
gives that, for $P > 1$ and $\alpha \in [0,1)$,
\begin{align}
& \Delta(P,P^{\alpha-1}) \nonumber \\
& \leq \frac{1}{2} \, \Biggl[ \min \biggl\{ \log(1 + P) + \log\left( 1 + \frac{P}{1+P^{\alpha}} \right), \;
2 \log\left( 1 + P^{\alpha} + \frac{P}{1 + P^{\alpha}} \right) \biggr\}
- \log\bigl(1+P+P^{\alpha}\bigr) \Biggr] \nonumber \\
& \triangleq \overline{\Delta}(P,P^{\alpha-1}).
\label{eq:upper bound on Delta for a specific setting}
\end{align}
Consequently, for $\alpha \in [0,1)$, we get from \eqref{eq:upper bound on Delta for a specific setting} that
\begin{align}
& \lim_{P \rightarrow \infty} \frac{\overline{\Delta}(P, P^{\alpha-1})}{\log P} \nonumber \\[0.1cm]
& = \frac{1}{2} \Bigl[ \min \Bigl\{2-\alpha, \;
2 \max\bigl\{\alpha, 1-\alpha\bigr\} \Bigr\} - 1 \Bigr] \nonumber \\[0.1cm]
& = \min \left\{ \frac{1-\alpha}{2}, \Bigl|\frac{1}{2}-\alpha \Bigr| \right\} \nonumber \\
& = \delta(\alpha)
\label{eq:asymptotic tightness of the upper bound on Delta}
\end{align}
where the last equality follows from \eqref{eq:closed-form expression for delta}.

The substitution of the cross-link gain $a=P^{\alpha-1}$ into the lower
bound on $\Delta(P,a)$ in
\eqref{eq:lower bound on Delta - symmetric Gaussian IC with weak interference}
gives that, for $P \geq 2.551$ and $\alpha \in [0,1)$,
\begin{align}
& \Delta(P,P^{\alpha-1}) \nonumber \\
& \geq \frac{1}{2} \, \Biggl[ \min \Biggl\{\log\bigl(1 + P + P^{\alpha}\bigr)
+  \log\left( 1 + \frac{P}{1 + P^{\alpha}} \right), \;
2 \log\left(1 + P^{\alpha} + \frac{P}{1 + P^{\alpha}} \right) \Biggr\} \nonumber \\
& \qquad - \min \left\{\log\bigl(1+P+P^{\alpha}\bigr) + \log\left(1+\frac{P}{(1+P^\alpha)^2} \right),
\log(1+2P) \right\} \Biggr] - 1 \nonumber \\
& \triangleq \underline{\Delta}(P,P^{\alpha-1}).
\label{eq:lower bound on Delta for the same specific setting}
\end{align}
Consequently, for $\alpha \in [0,1)$, it follows from
\eqref{eq:lower bound on Delta for the same specific setting} that
\begin{align}
& \lim_{P \rightarrow \infty} \frac{\underline{\Delta}(P, P^{\alpha-1})}{\log P} \nonumber \\
& \stackrel{\text{(a)}}{=} \frac{1}{2} \left[ \min \Bigl\{2-\alpha, \, 2 \max\{\alpha, 1-\alpha\}\Bigr\} -1 \right] \nonumber \\
& \stackrel{\text{(b)}}{=} \delta(\alpha)
\label{eq:asymptotic tightness of the lower bound on Delta}
\end{align}
where the substraction by~1 in equality~(a) follows from the satisfiability of the inequality
$$\log(1+P) \leq \min \left\{\log\bigl(1+P+P^{\alpha}\bigr) +
\log\left(1+\frac{P}{(1+P^\alpha)^2} \right), \log(1+2P) \right\} \leq \log(1+2P)$$
which therefore implies that
$$\lim_{P \rightarrow \infty} \frac{\min \left\{\log\bigl(1+P+P^{\alpha}\bigr) +
\log\left(1+\frac{P}{(1+P^\alpha)^2} \right), \log(1+2P) \right\}}{\log P} = 1.$$
Equality~(b) in \eqref{eq:asymptotic tightness of the lower bound on Delta}
follows from the last two equalities in \eqref{eq:asymptotic tightness of the upper bound on Delta}.

To conclude, \eqref{eq:asymptotic tightness of the upper bound on Delta} and
\eqref{eq:asymptotic tightness of the lower bound on Delta}
demonstrate the asymptotic tightness of the upper and lower bound in
\eqref{eq:upper bound on Delta - symmetric Gaussian IC with weak interference} and
\eqref{eq:lower bound on Delta - symmetric Gaussian IC with weak interference}, respectively,
for the considered coupling of the SNR and INR in \eqref{eq:scaling of SNR and INR}. This
completes the proof of the theorem.
\end{proof}

\begin{remark}
{\em The following is a discussion on
Theorem~\ref{theorem: asymptotic tightness of the bounds on the excess-rate}. Consider the case
where the SNR and INR scalings are coupled such that \eqref{eq:scaling of SNR and INR} holds.
Under this assumption, the reason for the asymptotic tightness of the upper and lower bounds
in \eqref{eq:upper bound on Delta for a specific setting} and
\eqref{eq:lower bound on Delta for the same specific setting} is twofold.
The first reason is related to the ETW bound that provides the exact asymptotic
linear growth of the sum-rate with $\log P$ (see \cite[Theorem~2]{EtkinTseWang08}).
The second reason is attributed to the fact that, for a weak and symmetric two-user GIC, the total
rate at the two corner points of the capacity region is
bounded between $\frac{1}{2} \, \log(1+P)$ and $\frac{1}{2} \, \log(1+2P)$ (see
\eqref{eq:bound1 on the corner points for weak interference}
and \eqref{eq:bound2 on the corner points for weak interference}) and both scale
like $\frac{1}{2} \, \log P$ for large~$P$. It is noted that an ignorance
of the effect of Kramer's bound in the derivation of
the upper bounds on the right-hand sides of
\eqref{eq:bound1 on the corner points for weak interference}
and \eqref{eq:bound2 on the corner points for weak interference} would have weakened
the lower bound on $\Delta(P,P^{\alpha-1})$ by a removal of
the term $\frac{1}{2} \, \log(1+2P)$ from the right-hand side of
\eqref{eq:lower bound on Delta for the same specific setting}. Consequently,
for $\alpha \in \bigl[0, \frac{1}{2}\bigr]$,
this removal would have reduced the asymptotic limit in
\eqref{eq:asymptotic tightness of the lower bound on Delta}
from $\delta(\alpha) = \frac{1}{2} - \alpha$ (see \eqref{eq:closed-form expression for delta})
to zero.}
\label{remark: explanation for the asymptotic tightness of the upper and lower bounds on Delta}
\end{remark}

As a consequence of the asymptotic analysis in this sub-section, some implications are provided
in the following:
\begin{enumerate}
\item Consider a two-user symmetric GIC where the cross-link gain is equal to $a = P^{\alpha-1}$
for $\alpha \geq 0$.
From \eqref{eq:definition of GDOF}, \eqref{eq:delta - analogous definition to GDOF},
\eqref{eq:closed-form expression for delta} and \eqref{eq:closed-form expression for the GDOF},
it follows that
\begin{align}
\lim_{P \rightarrow \infty} \frac{\Delta(P,P^{\alpha-1})}{C_{\text{sum}}(P,P^{\alpha-1})}
= \frac{\delta(\alpha)}{d(\alpha)} = \left\{ \begin{array}{cl}
    \frac{1-2\alpha}{2(1-\alpha)}, & \mbox{if $0 \leq \alpha < \frac{1}{2}$} \\[0.1cm]
    1-\frac{1}{2\alpha}, & \mbox{if $\frac{1}{2} \leq \alpha < \frac{2}{3}$} \\[0.1cm]
    \frac{1-\alpha}{2-\alpha}, & \mbox{if $\frac{2}{3} \leq \alpha < 1$} \\[0.1cm]
    0, & \mbox{if $\alpha \geq 1$}
  \end{array} \right.
  \label{eq:closed form expression for the normalized loss in total rate}
\end{align}
is the asymptotic fractional loss in the total rate at the corner points of the capacity
region.

\item Analogously to the GDOF of a two-user symmetric GIC in \eqref{eq:definition of GDOF},
the function $\delta$ is defined in \eqref{eq:delta - analogous definition to GDOF} by
replacing the sum-rate with the excess rate for the sum-rate w.r.t. the corner points;
in both cases, it is assumed that the cross-link gain is $a = P^{\alpha-1}$ for some
interference level $\alpha \geq 0$.
The GDOF is known to be a non-monotonic function of $\alpha$ over the interval [0,1]
(see \cite[pp.~5542--5543]{EtkinTseWang08} and \eqref{eq:closed-form expression for the GDOF}).
From \eqref{eq:closed-form expression for delta},
it also follows that $\delta$ is a non-monotonic function over this interval.
For $P > 1$, the cross-link gain $a = P^{\alpha-1}$
forms a monotonic increasing function of $\alpha \in [0, 1]$, and it is a one-to-one mapping
from the interval $[0,1]$ to itself.
This implies that, for large $P$, the excess rate for the sum-rate w.r.t. the corner
points (denoted by $\Delta(P,a)$) is a non-monotonic function of $a$ over the interval [0,1].
This observation is supported by numerical results in Section~\ref{subsection:numerical results}.
A discussion on this phenomenon is provided later in this section (see
Remark~\ref{remark: non-monotonicty of Delta for large P}).

\item Consider the closed-form expression in \eqref{eq:closed-form expression for the GDOF}
for the GDOF of a symmetric two-user GIC. For large $P$, the worst interference w.r.t. the sum-rate
is known to occur when the cross-link gain scales like $\frac{1}{\sqrt{P}}$ or it is~1 (this refers
to $\alpha = \frac{1}{2}$ or $\alpha=1$, respectively).
If $\alpha = \frac{1}{2}$, we have from \eqref{eq:closed-form expression for the GDOF}
\begin{align}
& d\Bigl(\frac{1}{2}\Bigr)
= \lim_{P \rightarrow \infty} \frac{C_{\text{sum}}\bigl(P, \frac{1}{\sqrt{P}}\bigr)}{\log P}
= \frac{1}{2}
\label{eq:GDOF for alpha equal to one-half}
\end{align}
and, from \eqref{eq:closed form expression for the normalized loss in total rate},
\begin{align}
\lim_{P \rightarrow \infty} \frac{\Delta\bigl(P,\frac{1}{\sqrt{P}}\bigr)}{C_{\text{sum}}\bigl(P,\frac{1}{\sqrt{P}}\bigr)} = 0.
\label{eq: an asymptotic limit for a that is the reciprocal of the square root of P}
\end{align}
The same also holds for the case where $\alpha=1$ (i.e., when the cross-link gain is $a=1$).
It therefore follows that, for the worst interference w.r.t. the sum-rate, there is asymptotically no
loss in the total rate $(R_1 + R_2)$ when the users operate at one of the corner points of the capacity region.

\item The limit on the left-hand side of \eqref{eq:closed form expression for the normalized loss in total rate}
is bounded between zero and one-half for $a = P^{\alpha-1}$ with $\alpha \geq 0$, and
it gets a local maximal value at $\alpha = \frac{2}{3}$ (which is global maximum for $\alpha \geq \frac{1}{2}$).
From Theorem~\ref{theorem:closed-form expression for delta}, we have
\begin{align}
\lim_{P \rightarrow \infty} \frac{\Delta\bigl(P,\frac{1}{\sqrt[3]{P}}\bigr)}{\log P} = \frac{1}{6}.
\label{eq: asymptotic limits for a that is the reciprocal of the 3rd root of P}
\end{align}

\item From the asymptotic upper and lower bounds on $\Delta(P,a)$ for large $P$ and a fixed $a \in (0,1)$
(see \eqref{eq:upper bound on Delta - large P, symmetric Gaussian IC with weak interference}
and \eqref{eq:lower bound on Delta - large P, symmetric Gaussian IC with weak interference}), we have
$$\frac{1}{2} \log\left(1+\frac{1}{a}\right)-1 \leq \lim_{P \rightarrow \infty} \Delta(P,a)
\leq \frac{1}{2} \log\left(\frac{1}{a}\right), \quad \forall \, a \in (0,1).$$
Since also the equality $\Delta(P,a)=0$ holds for every $a \geq 1$, then it follows that
$$\lim_{P \rightarrow \infty} \frac{\Delta(P,a)}{\log P} = 0, \quad \forall \, a > 0.$$
This is consistent with the equality $\delta(1)=0$ in \eqref{eq:closed-form expression for delta}.

\item Consider the capacity region of a weak and symmetric two-user GIC, and the bounds
on the excess rate for the sum-rate w.r.t. the corner points of its capacity region
(see Sections~\ref{subsection: An upper bound on delta for the corner points of the capacity region}
and~\ref{subsection: A lower bound on delta for the corner points of the capacity region}).
In this case, the transmission rate of one of the users is assumed to be equal to the single-user capacity of the
respective AWGN channel. Consider now the case where the transmission rate of this user is reduced
by no more than $\varepsilon > 0$, so it is within $\varepsilon$ of the single-user capacity.
Then, from Theorem~\ref{theorem: bounds on the corner points of a GIC with weak interference},
it follows that the upper bound on the transmission rate of the other user cannot increase by more than
$$f(\varepsilon) \triangleq \max\left\{2 \varepsilon, \, \left(1+\frac{1+P}{a P}\right) \varepsilon \right\}.$$
Consequently, the lower bound on the excess rate for the sum-rate in
\eqref{eq:lower bound on Delta - symmetric Gaussian IC with weak interference} is reduced by no more than
$f(\varepsilon)$. Furthermore, the upper bound on this excess rate cannot increase by more than $\varepsilon$
(note that if the first user reduces its transmission rate by no more than $\varepsilon$, then the other user
can stay at the same transmission rate; overall, the total transmission rate it decreased by no more than $\varepsilon$,
and consequently the excess rate for the sum-rate cannot increase by more than $\varepsilon$).
Revisiting the analysis in this sub-section by introducing a positive $\varepsilon \triangleq \varepsilon(P)$
to the calculations, before taking the limit of $P$ to infinity, leads to the conclusion that the corresponding
characterization of $\delta$ in \eqref{eq:closed-form expression for delta} stays un-affected as long as
$$\lim_{P \rightarrow \infty} \frac{\varepsilon(P)}{\log P} = 0$$
which then implies that $$\lim_{P \rightarrow \infty} \frac{f\bigl(\varepsilon(P)\bigr)}{\log P} = 0$$ when the value
of the cross-link gain $a$ is fixed.
For example, this happens to be the case if $\varepsilon$ scales like
$(\log P)^{\beta}$ for an arbitrary $\beta \in (0,1)$ (so, in the limit where
$P \rightarrow \infty$, we have $\varepsilon(P) \rightarrow \infty$ but
$\frac{\varepsilon(P)}{\log P} \rightarrow 0$).
\end{enumerate}

Consider a weak and symmetric GIC where, in standard form,
$P_1 = P_2 = P$ and $a_{12} = a_{21} = a \in (0,1)$.
Let $\Delta$ denote the excess rate for the sum-rate w.r.t. the corner points of
the capacity region, as it is defined in \eqref{eq:Delta}.
The following summarizes the results that are introduced in this section so far
for this channel model:
\begin{itemize}
\item The excess rate $\Delta$ satisfies the upper bound in
\eqref{eq:upper bound on Delta - symmetric Gaussian IC with weak interference}.
\item If $P \geq 2.551$, it also satisfies the lower bound in
\eqref{eq:lower bound on Delta - symmetric Gaussian IC with weak interference}.
\item For large enough $P$, $\Delta = \Delta(P,a)$ is a non-monotonic
function of $a$ over the interval $(0,1]$.
\item The upper and lower bounds on $\Delta(P, P^{\alpha-1})$ in
\eqref{eq:upper bound on Delta for a specific setting} and
\eqref{eq:lower bound on Delta for the same specific setting}, respectively,
imply the exact asymptotic scaling of $\Delta(P, P^{\alpha-1})$
with $\log P$ for an arbitrary $\alpha \geq 0$ (note that these bounds
apply to $\alpha \in [0,1)$, but $\Delta(P, P^{\alpha-1})=0$ when $\alpha \geq 1$ and $P \geq 1$).
\item The asymptotic linear growth of $\Delta(P, P^{\alpha-1})$ with $\log P$, for $\alpha \geq 0$,
is given by $\delta(\alpha)$ in \eqref{eq:closed-form expression for delta}.
Furthermore, a connection between the function $\delta$ and the symmetric GDOF is given in
\eqref{eq:connection between delta and GDOF for weak interference} (see Fig.~\ref{Figure:GDOF_and_delta}).
\item When the value of the cross-link gain is kept fixed between~0 and~1, the excess rate
$\Delta$ satisfies the upper and lower bounds in
\eqref{eq:upper bound on Delta - large P, symmetric Gaussian IC with weak interference} and
\eqref{eq:lower bound on Delta - large P, symmetric Gaussian IC with weak interference}, respectively.
These asymptotic bounds on $\Delta$ scale
like $\frac{1}{2} \, \log \left(\frac{1}{a}\right)$, and they
differ by at most 1~bit per channel use, irrespectively of the fixed value of
$a \in (0,1]$.
\item Let $a = P^{\alpha-1}$ for some $\alpha \geq 0$ and $P>1$.
Consider the loss in the total rate, expressed as a fraction of the sum-rate,
when the users operate at one of the corner points of the capacity region.
This asymptotic normalized loss is provided in
\eqref{eq:closed form expression for the normalized loss in total rate},
and it is bounded between 0 and $\frac{1}{2}$. For large values of $P$, it
roughly varies from 0 to $\frac{1}{4}$ by letting $a$ grow (only slightly)
from $\frac{1}{\sqrt{P}}$ to $\frac{1}{\sqrt[3]{P}}$.
\end{itemize}

The following remark refers to the third item above:
\begin{remark}
{\em For a weak and symmetric two-user GIC, the excess rate for the sum-rate w.r.t.
the corner points is the difference between the sum-rate of the capacity region
and the total rate at any of the two corner points of the capacity region. According
to Theorem~\ref{theorem: bounds on the corner points of a GIC with weak interference},
for large $P$, the total rate at a corner point is an increasing function of $a \in (0,1]$.
Although it is known that, for large $P$,
the sum-rate of the capacity region is not monotonic decreasing in $a$, a priori, there was a
possibility that by subtracting from it a monotonic increasing function in $a$, the difference
(that is equal to the excess rate $\Delta$) would be monotonic decreasing in $a$. However,
it is shown not to be the case. The fact that, for large $P$, the excess rate $\Delta(P,a)$ is not
a monotonic decreasing function of $a$ is a stronger property than the non-monotonicity
of the sum-rate.}
\label{remark: non-monotonicty of Delta for large P}
\end{remark}

\subsection{A Tightening of the Bounds on the Excess Rate $(\Delta)$ for Weak and Symmetric GICs}
\label{subsection:a tightening of the bounds on the excess rate for symmetric GICs with weak interference}
In Sections~\ref{subsection: An upper bound on delta for the corner points of the capacity region}
and~\ref{subsection: A lower bound on delta for the corner points of the capacity region}, closed-form
expressions for upper and lower bounds on $\Delta$ are derived for weak GICs.
These expressions are used in Section~\ref{subsection:An analogous measure to the generalized degrees of freedom}
for an asymptotic analysis where we let $P$ tend to infinity. In the following, the bounds
on the excess rate $\Delta$ are improved for finite $P$ at the cost of introducing bounds that are
subject to numerical optimizations. For simplicity, we focus on the model of a weak and symmetric GIC.
In light of Theorem~\ref{theorem: asymptotic tightness of the bounds on the excess-rate}, a use
of improved bounds does not imply any asymptotic improvement as compared to the
bounds in Section~\ref{subsection:An analogous measure to the generalized degrees of freedom}
that are expressed in closed form.
Nevertheless, the new bounds are improved for finite SNR and INR, as is illustrated in
Section~\ref{subsection:numerical results}.

\subsubsection{An improved lower bound on $\Delta$}
An improvement of the lower bound on the excess rate for the sum-rate w.r.t. the corner points ($\Delta$)
is obtained by relying on an improved lower bound on the sum-rate in comparison to
\eqref{eq:lower bound on Delta - symmetric Gaussian IC with weak interference}.
For tightening the lower bound on the sum-rate, it is suggested to combine
\eqref{eq:lower bound on Delta - symmetric Gaussian IC with weak interference} with the
lower bound in \cite[Eq.~(32)]{Sason04} (the latter bound follows
from the Han-Kobayashi achievable region, see \cite[Table~1]{Sason04}):
\begin{align}
& C_{\text{sum}} \geq \max_{\alpha, \beta, \delta} \, \rho(P, a, \alpha, \beta, \delta)
\label{eq:lower bound on the sum-rate - symmetric GIC with weak interference with 3-parameter optimization}
\end{align}
where
\begin{align*}
\rho(P,a, \alpha, \beta, \delta)
& \triangleq \delta \, \log \left( 1 + \frac{2 \alpha \delta P}{1+2a \beta \delta P} \right) +
\delta \, \log \left( 1 + \frac{2 \beta \delta P}{1+2a \alpha \delta P} \right)
+ \left(\frac{1-2 \delta}{2}\right) \, \log\bigl(1+2(1+2 \delta)P\bigr) \nonumber \\
&\quad + \min
\begin{aligned}[t]
\Biggl\{&\frac{\delta}{2} \cdot \log \left( 1 + \frac{2 \overline{\alpha}
\delta P + 2a \overline{\beta} \delta P}{1 + 2 \alpha
\delta P + 2a \beta \delta P} \right) +
\frac{\delta}{2} \cdot \log \left( 1 + \frac{2 \overline{\beta}
\delta P + 2a \overline{\alpha} \delta P}{1 + 2 \beta
\delta P + 2a \alpha \delta P} \right), \nonumber \\
&\delta \, \log
\left( 1 + \frac{2 \overline{\alpha} \delta P}{1 + 2 \alpha
\delta P + 2a \beta \delta P} \right) +
\delta \, \log \left( 1 + \frac{2 \overline{\beta}
\delta P}{1 + 2 \beta \delta P + 2a \alpha \delta P} \right), \\
&\delta \, \log
\left( 1 + \frac{2a \overline{\alpha} \delta P}{1 + 2a \alpha
\delta P + 2 \beta \delta P} \right) +
\delta \, \log \left( 1 + \frac{2a \overline{\beta}
\delta P}{1 + 2 \alpha \delta P + 2a \beta \delta P} \right)\Biggr\}
\end{aligned} \\[0.1cm]
& \qquad \, \forall \, (\alpha, \beta, \delta) \quad \text{s.t.}
\quad 0 \leq \alpha \leq 1, \quad 0 \leq \beta \leq 1, \quad 0 \leq \delta \leq \frac{1}{2} \, .
\end{align*}
A combination of \eqref{eq:Delta},
\eqref{eq:upper bound on total throughput for the two corner points - weak interference} and
\eqref{eq:lower bound on the sum-rate - symmetric GIC with weak interference with 3-parameter optimization}
gives the following lower bound on $\Delta$ for a weak and symmetric GIC:
\begin{align*}
\Delta &= \Delta(P,a)
\geq \max_{\alpha, \beta, \delta} \Bigl\{ \rho(P, a, \alpha, \beta, \delta) \Bigr\} -
\frac{1}{2} \, \min \biggl\{ \log\bigl(1+(a+1)P\bigr) + \log\left(1+\frac{P}{(1+aP)^2} \right),
\, \log(1+2P) \biggr\}.
\end{align*}
Furthermore, it follows from Lemma~\ref{lemma for a symmetric Gaussian IC} that if $P \geq 2.551$,
a combination of \eqref{eq:Delta} and
\eqref{eq:upper bound on total throughput for the two corner points - weak interference}
with the two lower bounds on the sum-rate
in \eqref{eq:lower bound on Delta - symmetric Gaussian IC with weak interference} and
\eqref{eq:lower bound on the sum-rate - symmetric GIC with weak interference with 3-parameter optimization}
gives the following tightened lower bound on $\Delta$ (as compared to
\eqref{eq:lower bound on Delta - symmetric Gaussian IC with weak interference}):
\begin{align}
\Delta \geq
& \max \Biggl\{ \max_{\alpha, \beta, \delta} \Bigl\{ \rho(P, a, \alpha, \beta, \delta) \Bigr\}, \; \nonumber \\
& \hspace*{1cm} \frac{1}{2} \, \min \biggl\{\log\bigl(1 + (a+1)P\bigr)  +  \log\left( 1 + \frac{P}{1 + aP} \right),
\; 2 \log\left(1 + aP + \frac{P}{1 + aP}\right) \biggr\} - 1 \Biggr\} \nonumber \\
& - \frac{1}{2} \, \min \biggl\{ \log\bigl(1+(a+1)P\bigr) + \log\left(1+\frac{P}{(1+aP)^2} \right),
\, \log(1+2P) \biggr\}.
\label{eq: tightened lower bound on Delta - symmetric Gaussian IC with weak interference}
\end{align}

\subsubsection{An improved upper bound on $\Delta$}
\label{subsubsection: An improved upper bound on Delta}
An improvement of the upper bound on the excess rate for the sum-rate w.r.t. the corner points ($\Delta$)
is obtained by relying on an improved upper bound on the sum-rate (as compared to
\eqref{eq:upper bound on the sum-rate - weak interference}).
This is obtained by calculating the minimum of Etkin's bound in \cite{Etkin_ISIT09}
and Kramer's bound in \cite[Theorem~2]{Kramer04}. Following the discussion in
\cite{Etkin_ISIT09}, Etkin's bound outperforms the upper bounds on the sum-rate in
\cite{Annapureddy_Veeravalli09}, \cite{EtkinTseWang08}, \cite{Motahari09}, \cite{Kramer09};
nevertheless, for values of $a$ that are close to~1, Kramer's bound in \cite[Theorem~2]{Kramer04}
outperforms the other known bounds on the sum-rate (see \cite[Fig.~1]{Etkin_ISIT09}).
Consequently, the minimum of Etkin's and Kramer's bounds in \cite{Etkin_ISIT09}
and \cite[Theorem~2]{Kramer04} is calculated as an upper bound on the sum-rate.
Combining \cite[Eqs.~(14)-(16)]{Etkin_ISIT09} (while adapting notation, and dividing the bound
by~2 for a real-valued GIC), the simplified version of Etkin's upper bound on the sum-rate
for real-valued, weak and symmetric GICs gets the form
\begin{align}
C_{\text{sum}} \leq
\min_{\alpha, \sigma, \rho} \Biggl\{ & \min \biggl\{
\frac{1}{2} \, \log\left(1+\frac{P(1+\alpha^2) \gamma}{(1-\rho^2) \sigma^2}\right),
\; \log\left(1+\frac{\alpha^2 P \gamma}{(1-\rho^2)\sigma^2}\right) \biggr\} \nonumber \\[0.1cm]
& + \log \left(\frac{\bigl(1+P(1+a)\bigr) \bigl(P(1+\alpha^2)+\sigma^2\bigr)-\bigl(P(1+\alpha \sqrt{a})+
\rho \sigma\bigr)^2}{P(1+\alpha^2)\gamma+(1-\rho^2)\sigma^2} \right) \Biggr\}
\label{eq:Etkin's upper bound on the sum-rate, ISIT09}
\end{align}
where
\begin{align}
\gamma = \alpha^2 - 2 \alpha \rho \sigma \sqrt{a} + \sigma^2 a, \quad
\rho = \alpha \sigma \sqrt{a} \pm \sqrt{(1-\alpha^2)(1-\sigma^2 a)}, \quad
\sigma \in \bigl[0, \frac{1}{\sqrt{a}}\bigr], \quad \alpha \in [-1, 1].
\label{eq: rho}
\end{align}
The two possible values of $\rho$ in \eqref{eq: rho} need to be checked in the optimization of
the parameters. For a weak and symmetric GIC, Kramer's upper bound on the sum-rate
(see \cite[Eqs.~(44) and~(45)]{Kramer04}) is simplified to
\begin{align}
C_{\text{sum}} \leq \frac{1}{2} \log\left(1+2P+\frac{B}{2}-
\frac{1}{2} \sqrt{B^2-4P^2 \left(\frac{1}{a}-1\right)^2}\right)
\label{eq:Kramer's upper bound on the sum-rate, IT 2004}
\end{align}
where
$B = \frac{1}{a^2}+2P \left(\frac{1}{a}-1\right)-1$. An improvement of the
upper bound on the sum-rate in \eqref{eq:upper bound on the sum-rate - weak interference}
follows by taking the minimal value of the bounds in
\eqref{eq:Etkin's upper bound on the sum-rate, ISIT09} and
\eqref{eq:Kramer's upper bound on the sum-rate, IT 2004}; consequently, a
combination of
\eqref{eq:Delta} and \eqref{eq:lower bound on total throughput for the two corner points - weak interference}
with this improved upper bound on the sum-rate
provides an improved upper bound on $\Delta$ (as compared to the bound in
\eqref{eq:upper bound on Delta - symmetric Gaussian IC with weak interference}).

\subsubsection{A simplification of the improved upper bound on $\Delta$ for a sub-class of weak and symmetric GICs}
\label{subsubsection: Simplification of the improved bounds on Delta for a sub-class of symmetric GICs}
The following simplifies the improved upper bound on the excess rate
$(\Delta)$ for a sub-class of weak and symmetric GICs. It has been independently
demonstrated in \cite{Annapureddy_Veeravalli09}, \cite{Motahari09} and \cite{Kramer09} that if
\begin{align}
0 < a < \frac{1}{4} \, , \quad 0 < P \leq \frac{\sqrt{a} - 2a}{2a^2}
\label{eq:conditions for symmetric GICs with weak interference where the sum-rate is known exactly}
\end{align}
then the sum-rate of the GIC is equal to
\begin{align}
C_{\text{sum}} = \log\left(1 + \frac{P}{1+aP}\right).
\label{eq:exact sum-rate for a sub-class of symmetric GICs with weak interference}
\end{align}
This sum-rate is achievable by using single-user Gaussian codebooks, and treating the
interference as noise. Under the conditions in
\eqref{eq:conditions for symmetric GICs with weak interference where the sum-rate is known exactly},
the exact sum-rate coincides with the upper bound given in
\eqref{eq:Etkin's upper bound on the sum-rate, ISIT09}. Hence,
a replacement of the upper bound on the sum-rate in
\eqref{eq:upper bound on the sum-rate - weak interference}
with the exact sum-rate in \eqref{eq:exact sum-rate for a sub-class of symmetric GICs with weak interference},
followed by a combination of \eqref{eq:Delta} and
\eqref{eq:lower bound on total throughput for the two corner points - weak interference} gives that
\begin{align}
\Delta \leq \frac{1}{2} \, \log\left(\frac{1}{1+aP} + \frac{P}{(1+aP)^2} \right).
\label{eq:tightened bound on Delta - sub-class of symmetric GICs with weak interference}
\end{align}
One can verify that, under the conditions in
\eqref{eq:conditions for symmetric GICs with weak interference where the sum-rate is known exactly},
the upper bound on $\Delta$ in
\eqref{eq:tightened bound on Delta - sub-class of symmetric GICs with weak interference} is indeed positive.

\subsection{Numerical Results}
\label{subsection:numerical results}
The following section presents numerical results for the bounds on the excess rate for the sum-rate
w.r.t. the corner points (denoted by $\Delta$) while focusing on weak and symmetric two-user GICs.

\begin{figure}[here!]
\begin{center}
\epsfig{file=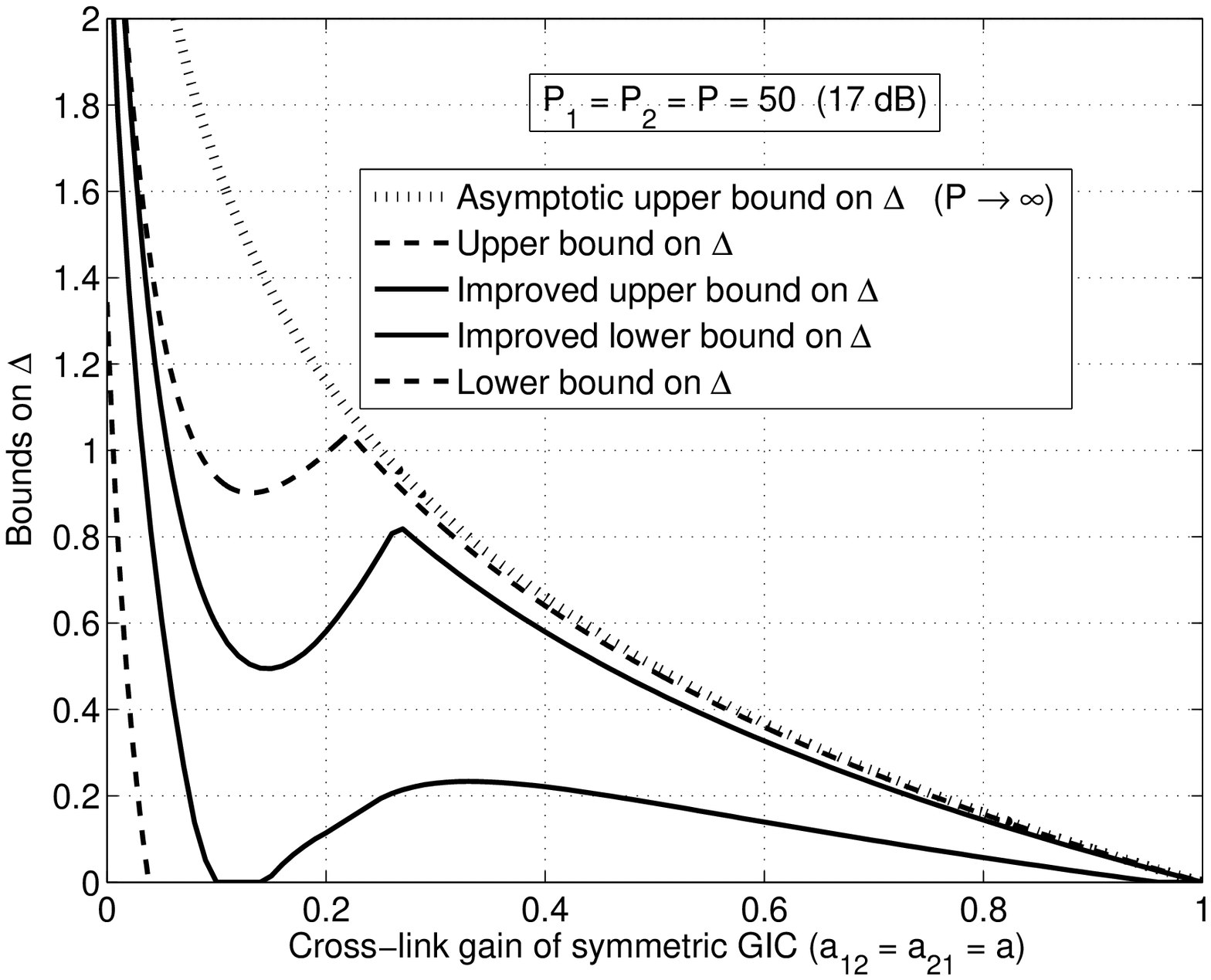,scale=0.6}\\
\epsfig{file=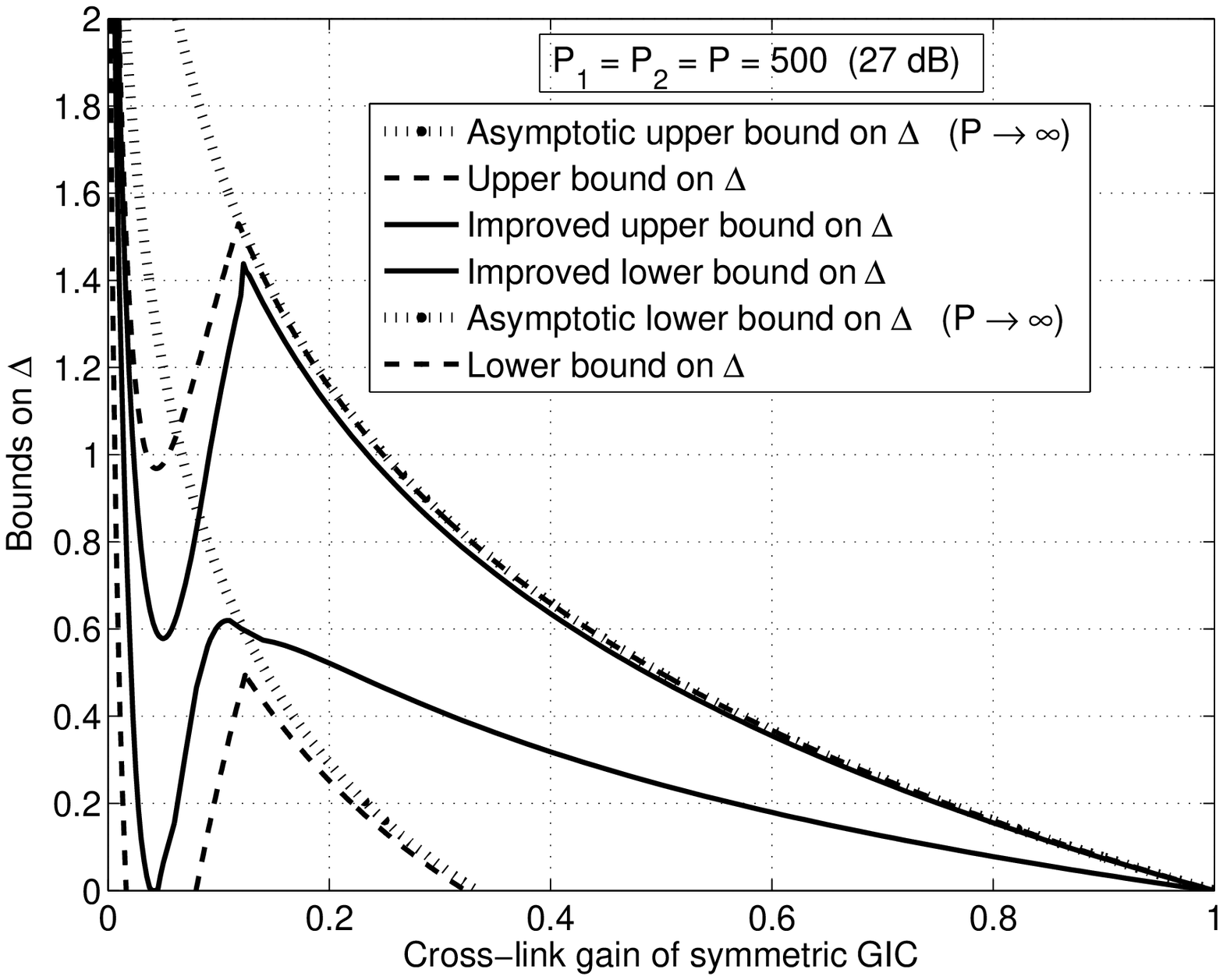,scale=0.6}
\caption{\label{Figure:Bounds on the excess rate - weak interference}
Upper and lower bounds on the excess rate for the sum-rate w.r.t. the corner points $(\Delta$) as a
function of the cross-link gain $(a)$. The plots refer to a weak and symmetric
GIC where $P_1 = P_2 = P$ and $a_{12} = a_{21} = a \in [0,1]$ in standard form.
The upper and lower plots refer to $P=50$ and $P=500$, respectively.
The upper and lower bounds on $\Delta$ rely on
\eqref{eq:upper bound on Delta - symmetric Gaussian IC with weak interference} and
\eqref{eq:lower bound on Delta - symmetric Gaussian IC with weak interference}, respectively,
and the improved bounds on $\Delta$ rely on
Section~\ref{subsection:a tightening of the bounds on the excess rate for symmetric GICs with weak interference}.
The dashed lines refer to the asymptotic upper and lower
bounds on $\Delta$ in \eqref{eq:upper bound on Delta - large P, symmetric Gaussian IC with weak interference}
and \eqref{eq:lower bound on Delta - large P, symmetric Gaussian IC with weak interference}, respectively.}
\end{center}
\end{figure}

\begin{figure}[here!]
\begin{center}
\epsfig{file=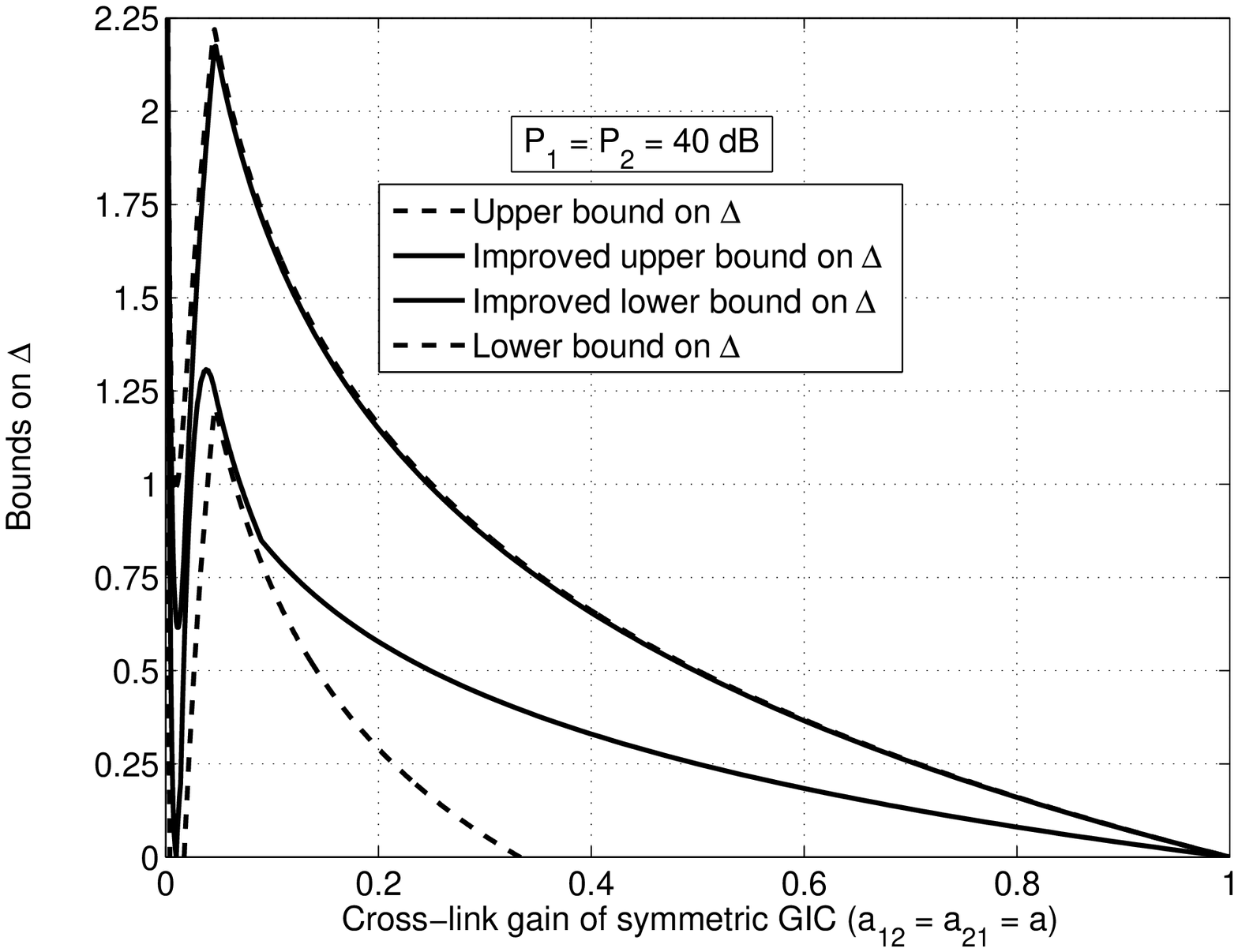,scale=0.6}\\
\epsfig{file=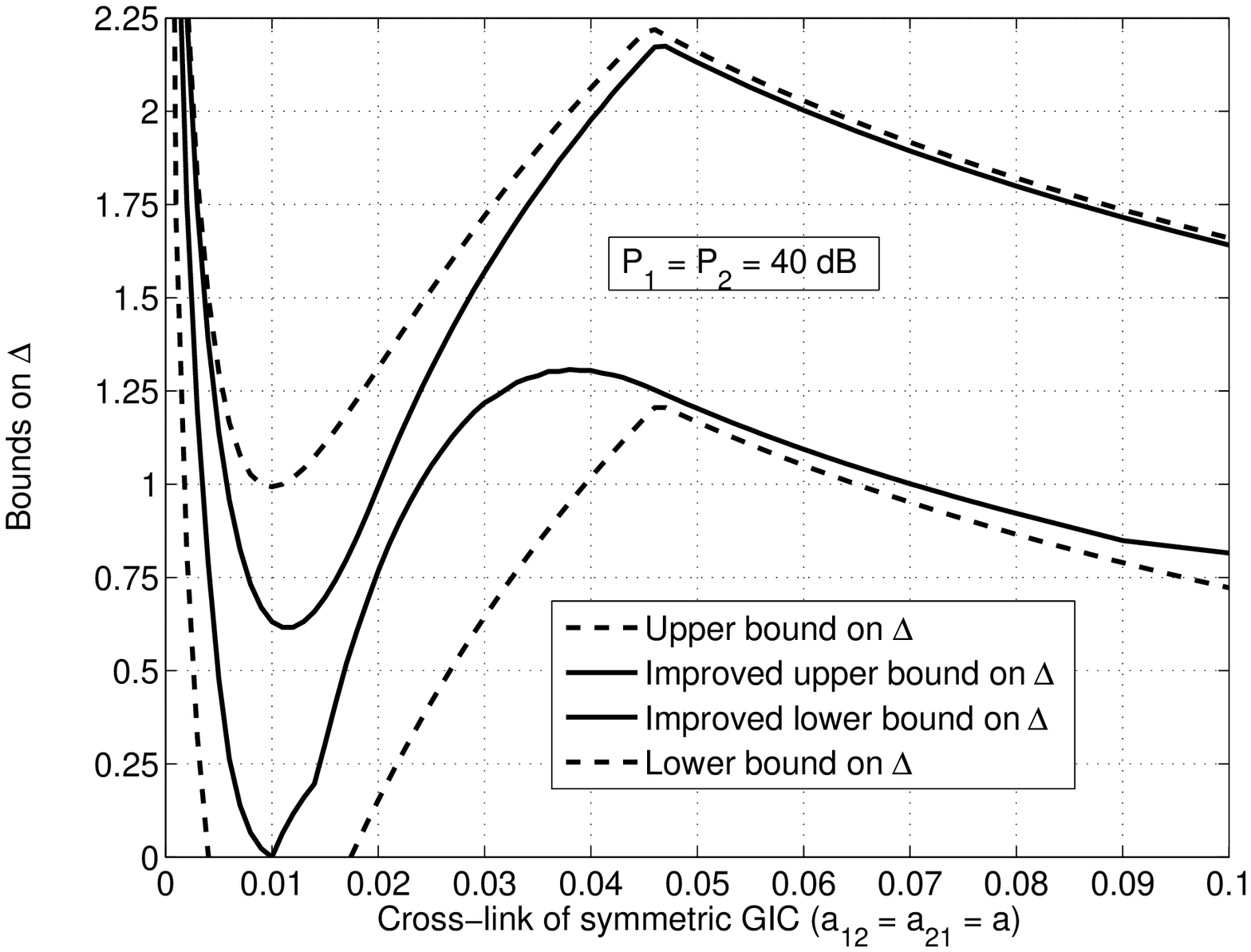,scale=0.6}
\caption{\label{Figure:Bounds on the excess rate - high P and weak interference}
Upper and lower bounds on the excess rate for the sum-rate w.r.t. the corner points $(\Delta$) as a
function of the cross-link gain $(a)$. This figure refers to a weak and symmetric
GIC where $P_1 = P_2 = P = 40$~dB and $a_{12} = a_{21} = a \in [0,1]$ in standard form.
The upper and lower bounds on $\Delta$ are given in
\eqref{eq:upper bound on Delta - symmetric Gaussian IC with weak interference} and
\eqref{eq:lower bound on Delta - symmetric Gaussian IC with weak interference}, respectively,
and the improved bounds on $\Delta$ rely on
Section~\ref{subsection:a tightening of the bounds on the excess rate for symmetric GICs with weak interference}.
The upper plot shows upper and lower bounds on $\Delta$ over the range of weak interference $(0 \leq a \leq 1$),
and the lower plot zooms in the upper plot for $a \in [0, 0.1]$; it shows that $\Delta$ is a non-monotonic function of $a$ in the weak interference regime.}
\end{center}
\end{figure}

Figure~\ref{Figure:Bounds on the excess rate - weak interference} compares upper and lower bounds
on $\Delta$ as a function of the cross-link gain for a weak and symmetric GIC.
The upper and lower plots of this figure correspond to $P=50$ and $P=500$, respectively. The upper
and lower bounds on $\Delta$ rely
on \eqref{eq:upper bound on Delta - symmetric Gaussian IC with weak interference} and
\eqref{eq:lower bound on Delta - symmetric Gaussian IC with weak interference}, respectively, and
the improved upper and lower bounds on $\Delta$ are based on
Section~\ref{subsection:a tightening of the bounds on the excess rate for symmetric GICs with weak interference}.
For $P=50$ (see the upper plot of Figure~\ref{Figure:Bounds on the excess rate - weak interference}), the
advantage of the improved bounds on $\Delta$ is exemplified; the lower bound on $\Delta$
for the case where $P=50$ is almost useless (it is zero unless the interference is very weak). The
improved upper and lower bounds on $\Delta$ for $P=50$ do not enable to conclude whether
$\Delta$ is a monotonic decreasing function of $a$ (for weak interference where $a \in [0, 1]$).
For $P=500$ (see the lower plot of Figure~\ref{Figure:Bounds on the excess rate - weak interference}),
the improved bounds on $\Delta$ indicate that it is not a monotonic decreasing function of $a$;
this follows by noticing that the improved upper bound on $\Delta$ at $a =0.045$ is equal to
0.578~bits per channel use, and its improved lower bound at $a = 0.110$ is equal to 0.620~bits per
channel use. The observation that, for large $P$, the function of $\Delta$ is not
monotonic decreasing in $a \in (0,1)$ is supported by the asymptotic analysis in
Section~\ref{subsection:An analogous measure to the generalized degrees of freedom}.
This conclusion is stronger than the observation that, for large enough $P$, the
sum-rate is not a monotonic decreasing function of $a \in [0,1]$ (see \cite[pp.~5542--5543]{EtkinTseWang08}),
as it is discussed in Remark~\ref{remark: non-monotonicty of Delta for large P} (see
Section~\ref{subsection:An analogous measure to the generalized degrees of freedom}).
Figures~\ref{Figure:Bounds on the excess rate - weak interference}
and~\ref{Figure:Bounds on the excess rate - high P and weak interference} show that the phenomenon of
the non-monotonicity of $\Delta$ as a function of $a$ is more dominant when the value of $P$ is increased.
These figures also illustrate the advantage of the improved upper and lower bounds on $\Delta$ in
Section~\ref{subsection:a tightening of the bounds on the excess rate for symmetric GICs with weak interference}
in comparison to the simple bounds on $\Delta$ in
\eqref{eq:upper bound on Delta - symmetric Gaussian IC with weak interference} and
\eqref{eq:lower bound on Delta - symmetric Gaussian IC with weak interference}. Note, however,
that the simple bounds on $\Delta$ that are given in closed-form expressions are asymptotically tight
as is demonstrated in Theorem~\ref{theorem: asymptotic tightness of the bounds on the excess-rate}.

\begin{table*}[here!]
\caption{Comparison of the asymptotic approximation of the excess rate for the sum-rate w.r.t. the corner points
$(\Delta)$ with its improved upper bound on $\Delta$ in
Section~\ref{subsubsection: An improved upper bound on Delta}.}
\label{Table: Numerical results for Delta} \centering
\begin{tabular}{|c||c|c||c|c|} \hline
Power constraint & Value of $a$ achieving &  Normalized $\Delta$
& Value of $a$ achieving & Normalized $\Delta$ \\
in standard form & minimum of $\Delta$ & by $\log P$
& maximum of $\Delta$ for $a \geq \frac{1}{\sqrt{P}}$
& by $\log P$ \\ \hline
& \; Asymptotic \; \; \; Exact & \; Asymptotic \; \; \; Exact
& \; Asymptotic \; \; \; Exact & \; Asymptotic \; \; \; Exact \\
& approximation \; \; value & approximation \; \; value
& approximation \; \; value & approximation \; \; value \\ \hline
$(P)$ & \hspace*{-1.2cm} $\bigl(a = \frac{1}{\sqrt{P}}\bigr)$
& \hspace*{-1.2cm} Eq.~\eqref{eq: an asymptotic limit for a that is the reciprocal of the square root of P}
& \hspace*{-1.2cm} $\bigl(a = \frac{1}{\sqrt[3]{P}}\bigr)$
& \hspace*{-1.2cm} Eq.~\eqref{eq: asymptotic limits for a that is the reciprocal of the 3rd root of P}
\\[0.1cm] \hline \hline
27~\text{dB} & \; \; \; 0.045 \; \; \; \; \; 0.050 & \; \; \; \; \; 0 \; \; \; \; \; \; 0.065
& \; \; \; 0.126 \; \; \; \; \; 0.140 & \; \; \; 0.167 \; \; \; \; \; 0.154 \\
40~\text{dB} & \; \; \; 0.010 \; \; \; \; \; 0.011 & \; \; \; \; \; 0 \; \; \; \; \; \; 0.046
& \; \; \; 0.046 \; \; \; \; \; 0.042 & \; \; \; 0.167 \; \; \; \; \; 0.164 \\
60~\text{dB} & \; \; \; 0.001 \; \; \; \; \; 0.001 & \; \; \; \; \; 0 \; \; \; \; \; \; 0.032
& \; \; \; 0.010 \; \; \; \; \; 0.010 & \; \; \; 0.167 \; \; \; \; \; 0.166 \\
\hline
\end{tabular}
\end{table*}

Table~\ref{Table: Numerical results for Delta} compares the asymptotic approximation of
$\Delta$ with its improved upper bound in Section~\ref{subsubsection: An improved upper bound on Delta}.
It verifies that, for large $P$, the minimal value of $\Delta$ is obtained at
$a \approx \frac{1}{\sqrt{P}}$; it also verifies that, for large $P$, the maximal value of $\Delta$
for $a \geq \frac{1}{\sqrt{P}}$ is obtained at $a \approx \frac{1}{\sqrt[3]{P}}$.
Table~\ref{Table: Numerical results for Delta} also supports the asymptotic limits in
\eqref{eq: an asymptotic limit for a that is the reciprocal of the square root of P}
and \eqref{eq: asymptotic limits for a that is the reciprocal of the 3rd root of P}, showing
how close are the numerical results for large $P$ to their corresponding asymptotic limits:
specifically, for large $P$, at $a = \frac{1}{\sqrt{P}}$ and $\frac{1}{\sqrt[3]{P}}$, the
ratio $\frac{\Delta}{\log P}$ tends to zero or~$\frac{1}{6}$, respectively; this is supported
by the numerical results in the 5th and 9th columns of Table~\ref{Table: Numerical results for Delta}.
The asymptotic approximations
in Table~\ref{Table: Numerical results for Delta} are consistent with the overshoots observed
in the plots of $\Delta$ when the cross-link gain $a$ varies between $\frac{1}{\sqrt{P}}$ and
$\frac{1}{\sqrt[3]{P}}$; this interval is narrowed as the value of $P$ is increased (see
Figures~\ref{Figure:Bounds on the excess rate - weak interference}
and~\ref{Figure:Bounds on the excess rate - high P and weak interference}).
Finally, it is also shown in Figures~\ref{Figure:Bounds on the excess rate - weak interference}
and~\ref{Figure:Bounds on the excess rate - high P and weak interference} that the
curves of the upper and lower bounds on $\Delta$, as a function of the cross-link gain
$a$, do not converge uniformly to their asymptotic upper and lower bounds in
\eqref{eq:upper bound on Delta - large P, symmetric Gaussian IC with weak interference}
and \eqref{eq:lower bound on Delta - large P, symmetric Gaussian IC with weak interference}, respectively.
This non-uniform convergence is noticed by the large deviation of the bounds for finite $P$ from the
asymptotic bounds where this deviation takes place over an interval of small values of $a$; however,
this interval of $a$ shrinks when the value of $P$ is increased, and its length is approximately $\frac{1}{\sqrt[3]{P}}$ for large $P$. This conclusion is consistent with the asymptotic analysis
in Section~\ref{subsection:An analogous measure to the generalized degrees of freedom}
(see the items that correspond to Eqs.~\eqref{eq:closed form expression for the normalized loss in total rate}
and~\eqref{eq: asymptotic limits for a that is the reciprocal of the 3rd root of P}), and it is also supported
by the numerical results in Table~\ref{Table: Numerical results for Delta}.

\section{Summary and Outlook}
\label{section: Summary and Outlook}
This paper considers the corner points of the capacity region of a two-user Gaussian
interference channel (GIC). The operational meaning of the corner points is a study of
the situation where one user sends its information at the single-user capacity
(in the absence of interference), and the other user transmits its data at the
largest rate for which reliable communication is possible at the two non-cooperating
receivers. The approach used in this work for the study of the corner points relies
on some existing outer bounds on the capacity region of a two-user GIC.

In contrast to strong, mixed or one-sided GICs, the two corner points of the capacity region
of a weak GIC have not been determined yet. This paper is focused
on the latter model that refers to a two-user GIC in standard form whose cross-link gains
are positive and below~1.
Theorem~\ref{theorem: bounds on the corner points of a GIC with weak interference}
provides rigorous bounds on the corner points of the capacity region, whose tightness
is especially remarkable at high SNR and INR.

The sum-rate of a GIC with either strong, mixed or one-sided interference is
attained at one of the corner points of the capacity region, and this corner
point is known exactly (see \cite{Han81}, \cite{Motahari09}, \cite{Sason04},
\cite{Sato81} and \cite{Kramer09}).
This is in contrast to a weak GIC whose sum-rate is not
attained at any of the corner points of its capacity region. This motivates
the study in
Section~\ref{section: On the Sub-Optimality of the Corner Points for a Gaussian IC with Weak Interference}
which introduces and analyzes the {\em excess rate for the sum-rate w.r.t.
the corner points}. This measure, denoted by $\Delta$, is defined to be the
gap between the sum-rate and the maximal total rate obtained by the two corner
points of the capacity region. Simple upper and lower bounds on $\Delta$ are
derived in
Section~\ref{section: On the Sub-Optimality of the Corner Points for a Gaussian IC with Weak Interference},
which are expressed in closed form, and the asymptotic characterization of these
bounds is analyzed. In the asymptotic case where the channel is interference limited
(i.e., $P \rightarrow \infty$) and symmetric, the corresponding upper and lower bounds
on $\Delta$ differ by at most 1~bit per channel use (irrespectively of the value of the
cross-link gain~$a$); in this case, both asymptotic bounds on $\Delta$ scale like
$\frac{1}{2} \, \log \left(\frac{1}{a}\right)$ for small $a$.

Analogously to the study of the generalized degrees of freedom (GDOF), an asymptotic
characterization of $\Delta$ is provided in this paper. More explicitly, under the
setting where the SNR and INR scalings are coupled such that
$\frac{\log(\text{SNR)}}{\log(\text{INR})} = \alpha$ for an arbitrary non-negative
$\alpha$, the exact asymptotic characterization of $\Delta$ is provided in
Theorem~\ref{theorem:closed-form expression for delta}.
Interestingly, the upper and lower bounds on $\Delta$ are demonstrated to be
{\em asymptotically tight for the whole range of this scaling} (see
Theorem~\ref{theorem: asymptotic tightness of the bounds on the excess-rate}).

For high SNR, the non-monotonicity of $\Delta$ as a function
of the cross-link gain follows from the asymptotic analysis, and
it is shown to be a stronger result than the non-monotonicity of
the sum-rate in \cite[Section~3]{EtkinTseWang08}.

Improved upper and lower bounds on $\Delta$ are introduced for
finite SNR and INR, and numerical results of these bounds are exemplified.
The numerical results in Section~\ref{subsection:numerical results}
verify the effectiveness of the approximations for high SNR that follow
from the asymptotic analysis of $\Delta$.

This paper supports in general
Conjecture~\ref{conjecture: Costa's conjecture for the corner points}
whose interpretation is that if one user transmits at its single-user
capacity, then the other user should decrease its rate such that both
decoders can reliably decode its message.

A recent work by Bustin {\em et al.} studied the corner points via the connection between
the minimum mean square error and mutual information \cite{Bustin_ISIT14},
providing another support (endorsement) to Costa's conjecture from a different perspective.

We list in the following some directions for further research that are
currently pursued by the author:
\begin{enumerate}
\item A possible tightening of the bound in \eqref{eq:maximal R1 for a two-user mixed Gaussian IC}
for a mixed GIC is of interest.
It is motivated by the fact that the upper bound for the corresponding corner point
is above the one in Conjecture~\ref{conjecture: Costa's conjecture for the corner points}.

\item The unknown corner point of a weak one-sided GIC satisfies the bounds
in Proposition~\ref{proposition: corner points for a one-sided Gaussian IC with weak interference};
it is given by $(R_1, C_2)$ where the gap between the upper and lower bounds on $R_1$ in
\eqref{eq: R1 for second corner point of Z-IC} is large for small values of $a$. An improvement
of these bounds is of interest (see the last paragraph in
Section~\ref{subsection: On Conjecture 1 for the one-sided Gaussian IC}).

\item A possible extension of this work to the class of semi-deterministic
interference channels in \cite{TelatarTse07}, which includes the two-user GICs
and the deterministic interference channels in \cite{ElGamalCosta82}.
\end{enumerate}

\subsection*{Acknowledgment}
I am grateful to Max H. M. Costa for recent stimulating discussions, and for personal
communications on this problem about~13 years ago. Feedback from Ronit Bustin, Max H. M.
Costa, Gerhard Kramer, Shlomo Shamai and Emre Telatar is acknowledged.
Emre Telatar is gratefully acknowledged for an interesting discussion.

\end{document}